\documentclass[11pt]{amsart}                                    
\usepackage{amsmath}
\usepackage{mdframed}
\usepackage{graphicx,amssymb}
\usepackage{algorithm}
\usepackage{algorithmic}
\usepackage{tcolorbox}
\usepackage{pdfsync}
\usepackage{amsthm,color}
\usepackage{enumitem}
\usepackage{multirow}
\usepackage{float}

\newtheorem{theorem}{Theorem}[section]
\newtheorem{lemma}[theorem]{Lemma}
\newtheorem{proposition}[theorem]{Proposition}
\newtheorem{definition}[theorem]{Definition}
\newtheorem{corollary}[theorem]{Corollary}

\newtheorem{remark}[theorem]{Remark}

\newtheorem{example}[theorem]{Example}

\def\R{{\mathbb R}}
\def\Z{{\mathbb Z}}
\def\Zd{{\Z}^d}
\def\Rd{{\R}^d}

\numberwithin{equation}{section}

\makeatletter
\def\BState{\State\hskip-\ALG@thistlm}
\makeatother

\begin{document}

\title[Phaseless Sampling and Reconstruction]{Phaseless Sampling and Reconstruction of Real-Valued Signals  in   Shift-Invariant Spaces}
\author{Cheng Cheng, Junzheng Jiang and  Qiyu Sun}
\address{Cheng: Department of Mathematics,
University of Central Florida,
Orlando 32816, Florida, USA}
\email{cheng.cheng@knights.ucf.edu}
\address{Jiang: School of Information and Communication, Guilin University of Electronic Technology, Guilin 541004, China}
\email{jzjiang@guet.edu.cn
}
\address{Sun: Department of Mathematics,
University of Central Florida,
Orlando 32816, Florida, USA.}
\email{qiyu.sun@ucf.edu}
\thanks{The project is partially supported by the National Science Foundation 
 (DMS-1412413).}

\maketitle

\begin{abstract}
Sampling in shift-invariant spaces  is a
realistic model for signals with smooth spectrum.
In this paper, we consider phaseless  sampling and reconstruction of real-valued signals in a shift-invariant space from their magnitude measurements on the whole Euclidean space and from their phaseless samples taken on a discrete set with finite sampling density.
We introduce an undirected graph to a signal 
 and use  connectivity of the graph to characterize whether the signal can be determined, up to a sign, from its magnitude measurements on the whole Euclidean space. Under the local complement property assumption on a shift-invariant space, we find a discrete set with finite sampling density such that  signals in the shift-invariant space, that
are determined from  their magnitude measurements on the whole Euclidean space,
can be reconstructed in a stable way from their phaseless samples taken on that discrete set. In this paper, we also propose a reconstruction algorithm  which provides a suboptimal approximation to the original signal when its noisy phaseless samples  are available only. Finally, numerical simulations are performed to demonstrate the robust
reconstruction of  box spline signals from their noisy phaseless samples.
\end{abstract}

\section{Introduction}
In this paper, we consider the  phaseless sampling and reconstruction problem whether a real-valued signal $f$ on $\Rd$ 
can be determined, up to a sign,  from its magnitude measurements $|f(x)|$ on $\Rd$ or a 
 subset $X\subset \Rd$.  The above problem is ill-posed  inherently
and it could be solved only if  we
 have some extra information  about the signal $f$.

The additional  knowledge about the signals   in this paper is that they live in a shift-invariant space
\begin{equation} \label{vphi.def} 
V(\phi):=\Big\{\sum_{{k}\in \Zd}
 c({k}) \phi({x}-{k}): \  c({ k})\in \R \ {\rm for \ all} \   { k} \in \Zd\Big\}\end{equation}
 generated by a real-valued  continuous function $\phi$ with compact support.
Shift-invariant spaces have been used in  wavelet analysis and approximation theory, and sampling in shift-invariant spaces
is a
realistic model for signals with smooth spectrum,
see \cite{AG01, AST05, Bow00, DDR94, jia92} and references therein. Typical examples of shift-invariant spaces include those generated by  refinable functions
(\cite{daubechiesbook, mallatbook}) and
 box splines  $M_\Xi$, which are defined by
\begin{equation}\label{boxspline.def}
\int_{\Rd}g({x})M_\Xi({x}) d{x}=\int_{\R^s} g(\Xi { y}) d{ y}, \ \ g\in L^2(\Rd),
\end{equation}
where  $\Xi 
\in \Z^{d\times s}$  is a matrix with full rank $d$ (\cite{deboorbook, unser99, wahba90}).

The phaseless sampling and reconstruction problem of one-dimensional signals in shift-invariant spaces has been studied in \cite{YC16, pyb14, volker14, shenoy16, T11}.  Thakur  proved in \cite{T11} that one-dimensional real-valued signals in a Paley-Wiener space,  the shift-invariant space generated by the sinc function  $\frac{\sin \pi t}{\pi t}$, could be reconstructed from their phaseless samples  taken
 at more than twice the Nyquist rate.
Reconstruction of one-dimensional signals in a
shift-invariant space
 was studied
in \cite{shenoy16} when frequency magnitude  measurements are available. Not all signals in a shift-invariant space generated by a  compactly supported function are determined, up to a sign, from their magnitude  measurements on the whole line. In \cite{YC16}, the set of signals that can be determined from their magnitude measurement on the real line $\R$ is fully characterized and a fast algorithm is proposed to
reconstruct signals in a shift-invariant space from  their phaseless samples taken on a discrete set with finite sampling density. Up to our knowledge, there is no literature available on the phaseless sampling and reconstruction of high-dimensional signals in a shift-invariant space, which is the core of this paper.

The phaseless sampling and reconstruction of signals in  a shift-invariant space
 is an infinite-dimensional phase retrieval problem, which has received considerable attention in recent years  \cite{daubechies16,  Rima16, alaifari16, cahill15,  YC16, mallat14, pyb14, volker14,    shenoy16, T11}. 
Phase retrieval  plays important roles in signal/speech/image processing and it is a highly nonlinear mathematical problem. Even in the finite-dimensional setting, there are still lots of  mathematical and engineering problems about phase retrieval unanswered.  The reader may refer to  \cite{BCE06, BBCE09,  CSV12, candes13, casazza17,  gao16, jaganathany15, liieee17, schechtman15} and references therein for historical remarks and recent advances.

\smallskip

 The paper is organized as follows.  An introductory  problem about phaseless sampling and reconstruction  in the shift-invariant space  $V(\phi)$ is that
 a  real-valued signal $f$ 
  is determined, up to a sign,  from its magnitude $|f(x)|, x\in \Rd$. An equivalence has been provided in \cite{YC16}, see
 Theorem \ref{separable.tm} in Section \ref{NScond.sec},  that the signal
 $f$ must be nonseparable,  i.e.,   $f$ is not the sum of two nonzero signals in $V(\phi)$
with their supports being essentially disjoint.  
 A natural question arisen
 is  how to  determine nonseparability of a given signal in a  shift-invariant space.
For a one-dimensional signal $f$,
 it is shown in \cite{YC16}
 that $f$ is nonseparable if and only if the amplitude vector  $(c(k))_{k\in \Z}$ does not have consecutive zeros. However, there is no corresponding notion of consecutive zeros in the high-dimensional shift-invariant spaces $V(\phi)$  with $d\ge 2$.
  In Section \ref{NScond.sec}, we introduce
 an undirected graph  ${\mathcal G}_f$ to a high-dimensional signal $f$ in the shift-invariant space $V(\phi)$ 
 and use the connectivity of the graph ${\mathcal G}_f$ to characterize the nonseparability of the signal $f$, i.e.,  it is determined, up to a sign, from the
 magnitude measurements $|f(x)|, x\in \Rd$, see Theorems \ref{necessarypr.thm} and \ref{sufficient.thm}.

 In Section \ref{finiteset.sec},
 we consider the preparatory problem
whether a signal in a shift-invariant space  is determined, up to a
sign, from its phaseless samples taken on a discrete set with finite sampling
density. A necessary condition is that the signal is nonseparable. In Theorem \ref{tofinite.thm}, we show that
the above nonseparable requirement is also sufficient, and furthermore the discrete sampling set can be selected explicitly to be of the form $\Gamma+\Zd$, where $\Gamma$ contains finitely many elements.
However, the above result does not provide us an algorithm to reconstruct
a nonseparable signal from its phaseless samples taken on that discrete set.
In
Section \ref{finiteset.sec}, we introduce a local complement property for the shift-invariant space $V(\phi)$,
 which is similar to the complement property for frames in Hilbert/Banach spaces \cite{alaifari16, BCE06, Bandeira14, cahill15, YC16}.
 Local complement property is closely related to local phase retrievability, see Appendix  \ref{localcomplementproperty.appendix}.
 We apply the local complement property in Theorem \ref{necandsuf.thm} to find another discrete set
with finite sampling density on which nonseparable signals in a shift-invariant space
can be recovered from their phaseless samples.
Moreover,
  our proof of Theorem \ref{necandsuf.thm} is constructive.

 Stability is a pivotal problem in the phaseless sampling and reconstruction.
We aim to find  a good approximation to the original signal $f$, up to a sign,
when its noisy  phaseless samples
 \begin{equation*}
 z_\epsilon(y)=|f(y)|+\epsilon(y), \ y\in X, \end{equation*}
 taken on a discrete set $X$ with finite sampling density are available only,
 where  $\epsilon=(\epsilon(y))_{y\in X}$ is the additive bounded noise.
 A conventional approach  of the above problem is to solve the following
 minimization problem,
\begin{equation*} 
\min \sum_{y\in X} \big ||g(y)|-
z_{ \epsilon}(y)\big|^2 \ {\rm subject \ to} \ g\in V(\phi).
\end{equation*}
However, it is infinite-dimensional and infeasible.
For a shift-invariant sampling set $X$ of the form $\bigcup_{m=1}^M \Gamma_m+\Zd$,
we propose a four-step approach starting from local minimization problems,
\begin{equation*} 
\min \sum_{y\in \Gamma_m+l} \big ||g(y)|-
z_{\epsilon}(y)\big|^2  \ {\rm subject \ to} \ g\in V(\phi),
\end{equation*}
where $1\le m\le M$ and $l\in \Zd$, see \eqref{meps.def1}--\eqref{thresholding} in Section \ref{reconstruction.section}.
 In Theorem \ref{stability.thm}, we establish stability of the four-step approach to reconstruct nonseparable signals in a shift-invariant space.
     The above stability
   implies the nonexistence of   
    resonance phenomenon
  when the noise level is far below  the minimal magnitude of amplitude vector of the original signal, see Remark \ref{remark5.2}.

A fundamental problem in the phaseless sampling and reconstruction is to
 design efficient and robust algorithms for signal reconstruction in a noisy environment.
 Based on the four-step approach in Section  \ref{reconstruction.section}, we
 propose  
 an algorithm  
to reconstruct nonseparable signals in $V(\phi)$ from their noisy phaseless samples on the shift-invariant set $\bigcup_{m=1}^M\Gamma_m+\Zd$.  The complexity of the proposed algorithm
depends almost linearly on the support length of the original nonseparable signal.
 The reader may refer to
\cite{CSV12, candes13, F78, hand2016, iwen16, schechtman15} and references therein 
for various algorithms to reconstruct a finite-dimensional signals from
magnitude  of  its finite frame measurements.
The implementation and performance of the proposed algorithm to recover box spline signals are given in
Section \ref{numerical.sec}.

 Proofs 
  are collected in Section \ref{proof.section}
  and the local complement property is
discussed in  Appendix \ref{localcomplementproperty.appendix}.

 Notation:  Denote  the cardinality of a set $E$  by $\#E$ and    the closed ball in $\Rd$ with center $x$ and radius $R\ge 0$  by $B(x, R)$. 
 Define the power $x^k=\prod_{i=1}^d x_i^{k_i}$
  for  $x=(x_1,\ldots, x_d)^T\in \Rd$ and $k=(k_1, \ldots, k_d)^T\in \Zd$, and  the partial order
 $x\le y$  for $y=(y_1, \ldots, y_d)^T\in \Rd$ if $x_i\le y_i, 1\le i\le d$.

\section{Phase retrievability, nonseparability and connectivity}\label{NScond.sec}

 The phase retrievability of a  real-valued signal on  $\Rd$ is  whether it is determined, up to a sign, from its magnitude measurements.  It is characterized in \cite{YC16} as follows.

 \begin{theorem}\label{separable.tm}
 Let $\phi$ be a real-valued continuous function with compact support,  and
  $V(\phi)$  be the shift-invariant space  in \eqref{vphi.def} generated by $\phi$.
Then a  signal $f\in V(\phi)$ is determined, up to a sign,  by its magnitude measurements $|f({x})|, {x}\in \Rd$, if and only if $f$ is nonseparable, i.e.,
there does not  exist nonzero  signals  $f_1$ and $f_2$ in $V(\phi)$ such that
\begin{equation}\label{separable.condition}
f=f_1+f_2\ \ {\rm and}\  \ f_1 f_2=0.\end{equation}
\end{theorem}

\smallskip

The question arisen  
 is  how to  determine nonseparability of a signal in a  shift-invariant space.
   To answer the above question, we need the one-to-one correspondence
         between an amplitude vector $c$ and
                  a signal $f$ in the shift-invariant space $ V(\phi)$,
   \begin{equation}
    \label{gli.def}
    c:=(c(k))_{k\in \Zd} \longmapsto \sum_{k\in \Zd} c(k)\phi(\cdot-k)=:f\in V(\phi),
    \end{equation}
which is known as
      the {\em global linear independence} of the generator $\phi$  \cite{ron90, jia92, ron89}. 
For $d=1$, the nonseparability of a signal in  a shift-invariant space is characterized in \cite{YC16}
 that its amplitude vector  does not have consecutive zeros. However, there is no corresponding notion of consecutive zeros in the high-dimensional setting $(d\ge 2)$. 
To characterize the nonseparability of signals  on $\Rd$ with $d\ge 2$, we 
introduce an undirected graph for a signal in the shift-invariant space $V(\phi)$
generated by a real-valued continuous function $\phi$ with compact support.

\begin{definition} {\rm
For any $f(x)=\sum_{{k}\in \Zd} c({k}) \phi({x}-{k})\in V(\phi)$,
define an {\em undirected graph}
\begin{equation}\label{graphG.def}
{\mathcal G}_f:=(V_f, E_f),\end{equation}
where
 the vertex set
 $$V_f=\{{k}\in \Zd: \ c({k})\neq  0\}$$
  contains  supports of  the amplitude vector of the signal $f$,  and
\begin{equation*} E_f = \big\{({k},{k}')\in  V_f\times V_f: \  \ { k}\neq {k}' \  {\rm and} \  \phi({x}-{k})
\phi({x}-{k}') \neq 0 \ {\rm for \  some} \ {x}\in  \Rd\big\}
\end{equation*}
is the edge set associated with the signal $f$.
}
\end{definition}

The graph ${\mathcal G}_f$  in \eqref{graphG.def} is well-defined for any signal $f$ in the shift-invariant space $V(\phi)$ when $\phi$ has global linear independence.
Moreover, 
\begin{equation}
(k,k')\in E_f \ {\rm if\ and\ only\ if} \
k-k'\in \Lambda_\phi,
\end{equation}
 where
$\Lambda_\phi$ contains all $k\in \Zd$ such that
\begin{equation}\label{necandsuf.thm.eq2}
S_k:=\{{x}: \ \phi({x}) \phi({x}-k)\ne 0\}\ne \emptyset.
\end{equation}
In the following theorem, we show that
  connectivity  of  the graph ${\mathcal G}_f$
is a necessary condition for the
 nonseparability of the signal $f\in V(\phi)$.

\begin{theorem}\label{necessarypr.thm}  Let $\phi$ be
  a compactly supported  continuous function on $\R^d$ with
global linear independence, and $V(\phi)$ be the shift-invariant space
\eqref{vphi.def} generated by $\phi$. If  $f\in V(\phi)$
 is nonseparable,
then the  graph ${\mathcal G}_f$ in \eqref{graphG.def} is connected.
\end{theorem}

%
 Before stating sufficiency for the connectivity  of  the graph ${\mathcal G}_f$,
we recall a concept of
 local linear independence on an open set.

\begin{definition}  {\rm
 Let $\phi$ be a continuous function with compact support and $A$ be an open set. We say that $\phi$ has {\em local linear independence on  $A$ } if  $\sum_{{ k}\in \Zd} c({ k})\phi({ x}-{k}) =0$ for all $x\in A$ implies that $ c({k})=0$
for all ${k}\in \Zd$  satisfying $\phi({ x}-{k})\not\equiv  0$ on $A$.
} \end{definition}

 The global linear independence of a compactly supported function  $\phi$
can be interpreted as its local linear independence  on $\Rd$ (\cite{ron90, sun93}). 
Define
 \begin{equation}\label{phiA.def} \Phi_A({x}):=\big(\phi({ x}-{ k})\big)_{{k}\in K_A},\ {x}\in A\end{equation}
 and
 \begin{equation}\label{KA.def}
K_A:=\{{k}\in \Z^d: \ \phi(\cdot-{k})\not\equiv  0\ {\rm  on} \ A\}. \end{equation}
One may 
verify that $\phi$ has  local linear independence on  $A$
 if and only if the dimension of the
linear space spanned by  $\Phi_A({x}), x\in A$,
is  the cardinality of the set $K_A$.  The above characterization  can be used to verify the local linear independence on a bounded open set,
especially when $\phi$ has the explicit expression. For instance, one may verify that the generator $\phi_0$ in  Example \ref{hatexample1} below has local linear independence on $(0, 1)$, but it is locally linearly
dependent on $(0, 1/2)$ and $(1/2, 1)$.

The local linear independence on any open sets and global linear independence are equivalent for some compactly supported functions, such as box splines and one-dimensional refinable functions (\cite{deboor83, dahmen85,  dahmen83, jia85,  S10}). In the following theorem, we show that the converse in Theorem \ref{necessarypr.thm} is also true
 if  the generator $\phi$ is assumed to have local linear independence on any open set.

\begin{theorem}\label{sufficient.thm}
Let $\phi$ be a compactly supported  continuous function on $\R^d$ with  local linear independence on any open set, and  $f$ be a signal in the shift-invariant space $V(\phi)$.
 If  the graph ${\mathcal G}_f$ in \eqref{graphG.def} is connected, then  $f$ is nonseparable. 
\end{theorem}

For $d=1$, we have
\begin{equation}\label{onedimensionalconnected}
(k,k')\in E_f \ {\rm if\ and\ only\ if} \
|k-k'|\le L-1,
\end{equation}
provided that
 the support of $\phi$
is   $[0, L]$  for some $L\ge 1$.
This together with Theorems
\ref{necessarypr.thm} and \ref{sufficient.thm} leads to the following result, which is established in \cite{YC16} under different assumptions on the generator $\phi$.

\begin{corollary}
Let $\phi$ be a compactly supported  continuous function on $\R$, and  $f=\sum_{k\in \Z} c(k) \phi(\cdot-k)$ be a signal in the shift-invariant space $V(\phi)$.
If $\phi$ has local linear independence on any open set and its supporting set is $[0,L]$ for some $L\ge 1$, then $f$ is nonseparable if and only if $\sum_{l=0}^{L-2} |c(k+l)|^2\ne 0$
 for all $K_-(f)-L+1<k<K_+(f)+1$, where
$K_-(f)=\inf\{k: c(k)\ne 0\}$ and $K_+(f)=\sup\{k: c(k)\ne 0\}$.
\end{corollary}

 As demonstrated by the following example,  the  connectivity of
 the graph ${\mathcal G}_f$  is not sufficient for the signal $f$ to be nonseparable
 if the local  linear independence assumption on the  generator
$\phi$ is dropped.

\begin{example}\label{hatexample1} {\rm
Define
 $\phi_0(t)= h(4t-1)+h(4t-3)+h(4t-5)-h(4t-7)$, where $h(t)=\max(1-|t|, 0)$ is the hat function supported on $[-1, 1]$. One may easily verify that $\phi_0$ is a continuous  function having global linear independence.
Set
 $$f_1(t)=\sum_{k\in \Z} \phi_0(t-k)\ \ {\rm and}\ \  f_2(t)=\sum_{k\in \Z} (-1)^k \phi_0(t-k).$$
Then $f_1$ and $f_2$ are nonzero signals in $V(\phi_0)$ supported on $[0, 1/2]+\Z$ and $[1/2, 1]+\Z$ respectively, and
$f_1(t)f_2(t)=0$ for all   $t\in \R$.
Hence $f_1\pm 2 f_2$
 have the same magnitude measurement $|f_1|+2|f_2|$ on the real line 
 but they are different, even up to a sign, i.e., $f_1+2f_2\not\equiv \pm (f_1-2f_2)$.
   On the other hand, one may 
   verify that
  their associated graphs
 ${\mathcal G}_{f_1\pm 2f_2}$ 
 are connected.
}\end{example}

 Consider  a continuous solution
 $\phi$ of a refinement equation
\begin{equation}\label{refinement.def}
\phi(x)=\sum_{n=0}^N a(n) \phi(2x-n)\ \  {\rm and}  \ \  \int_{\R} \phi(x)dx=1 \end{equation}
with global linear independence,  where $\sum_{n=0}^N a(n)=2$  and $N\ge 1$ (\cite{daubechiesbook, mallatbook}).  The B-spline $ 
B_N$ of order $N$, which is  obtained
by convolving the indicator function $\chi_{[0,1)}$ on the unit interval  $N$ times,
satisfies the above refinement equation
(\cite{unser99, wahba90}).
The function $\phi$  in \eqref{refinement.def}
has support $[0, N]$ and
 it  has local linear independence on any open set if and only if it has global linear independence
 (\cite{lemarie91, meyer91, 
  S10}). Therefore we  have the following result for wavelet signals by Theorems \ref{necessarypr.thm} and \ref{sufficient.thm}, which is also established in
 \cite{YC16} with a different approach.

 \begin{corollary}\label{wavelet.cor}
 Let $\phi$ satisfy the refinement equation \eqref{refinement.def} and have global linear independence. Then $f\in V(\phi)$ is nonseparable
         if and only if the graph ${\mathcal G}_f$ in \eqref{graphG.def} is connected.
 \end{corollary}

 The  local linear independence requirement in Theorem \ref{sufficient.thm} can be verified
for  box splines  $M_\Xi$ in \eqref{boxspline.def}.
It is known that the box spline $M_\Xi$  has local linear independence on any open set if and only if
all $d\times d$ submatrices  of $\Xi$ have   determinants being either $0$ or $\pm 1$
if and only if it
has  global linear independence
(\cite{deboor83,  dahmen85,  dahmen83, jia85}).
   The reader may refer to \cite{deboorbook}
 for more properties and applications of  box splines.
As applications of Theorems \ref{necessarypr.thm} and \ref{sufficient.thm}, we have the following result for  box spline signals.

 \begin{corollary}\label{spline222.cor}
 Let $\Xi 
 \in \Z^{d\times s}$
  be a  matrix of full rank $d$ such that  its $d\times d$ submatrices  have   determinants being either $0$ or $\pm 1$.
 Then $f\in V(M_\Xi)$ is nonseparable
if and only if the graph ${\mathcal G}_f$ in \eqref{graphG.def} is connected.
 \end{corollary}


\section{Phaseless sampling and reconstruction}\label{finiteset.sec}

In this section, we consider the problem
whether a signal
in the shift-invariant space $V(\phi)$ is determined, up to a sign, from its phaseless samples taken on a discrete  set with finite sampling density.
Here we define the sampling density of a discrete
  set $X\subset \Rd$ by
  $$D(X):=D_+(X)=D_-(X)$$
if its    upper  sampling density $D_+(X)$ and
lower sampling density $D_-(X)$ are the same \cite{AG01,  cjs16, sun06}, where
\begin{equation}\label{samplingrate.def}
D_+(X):=\limsup_{R\to +\infty}\sup_{{x}\in \Rd}   \frac{\#(X\cap B({ x}, R))}{R^d} 
\end{equation}
and
\begin{equation}\label{samplingrate.def2}
D_-(X):=\liminf_{R\to +\infty}\inf_{{x}\in \Rd}   \frac{\#(X\cap
 B({x}, R))}{R^d}. 
\end{equation}
One may easily verify that a shift-invariant set $X=\Gamma+\Zd$ generated by a finite set $\Gamma$ has sampling density $\# \Gamma$.

 To determine a signal, up to a sign, from its phaseless samples taken on a discrete set, a necessary condition is that  the signal is nonseparable (hence phase retrievable). 
  In the next theorem, we show that the above requirement is also sufficient.

 \begin{theorem}\label{tofinite.thm}  Let $\phi$ be a compactly supported continuous function
  and $V(\phi)$ be the shift-invariant space in \eqref{vphi.def} generated by  $\phi$. Then there exists a discrete set $\Gamma\subset (0, 1)^d$ such that
any nonseparable signal $f\in  V(\phi)$ is determined, up to a sign,
 from its phaseless samples on the set $\Gamma+\Zd$ with finite sampling density.
 \end{theorem}

For a compactly supported function $\phi$ and a bounded open set $A$,
let
\begin{equation}\label{WPhiA.def}
W_{ A}\ {\rm  be \ the\ linear \ space\ spanned \ by} \ \Phi_{A}({x}) (\Phi_{A}({x}))^T,\ x\in A,
\end{equation}
where $\Phi_A$ is given in \eqref{phiA.def}.
Observe that for any bounded set $A$, the space $W_A$   spanned by
outer products  $\Phi_{A}({x}) (\Phi_{A}({x}))^T, x\in A$,
is of finite dimension.  Therefore there exists a finite set $\Gamma \subset A$
such that
outer products
$\Phi_{A}({\gamma}) (\Phi_{A}({\gamma}))^T, \gamma\in \Gamma$,
is a basis   of the linear space  $W_{A}$.
In the proof of Theorem \ref{tofinite.thm}, we use $A=(0, 1)^d$ and apply the above procedure to select
 the finite set $\Gamma$. With the above selection of the set $\Gamma$,
 \begin{equation}\label{dimintofinite.eq}
 \# \Gamma=\dim  W_{(0,1)^d},\end{equation}
 and
 $|f(x)|^2, x\in \Rd$, are determined by $|f(\gamma)|^2, \gamma \in \Gamma+\Zd$, see Section \ref{tofinite.section} for the detailed  proof.

 As symmetric matrices in  the space $W_{(0, 1)^d}$ are of size $\# K_{(0,1)^d}$, we have the following result  about the sampling density.

\begin{corollary}\label{corollary32} Let  $\phi$ and $V(\phi)$ be as in Theorem
\ref{tofinite.thm}. Then
 any nonseparable signal $f\in V(\phi)$  is determined from its phaseless samples on
a shift-invariant set $\Gamma+\Zd$ with sampling density
$$D(\Gamma+\Zd)\le \dim  W_{(0,1)^d}\le \frac{1}{2} \# K_{(0,1)^d} (\#K_{(0,1)^d}+1),$$
where $K_{(0,1)^d}$ is in \eqref{KA.def}.
\end{corollary}

The explicit construction of a discrete set  with finite sampling density in Theorem
\ref{tofinite.thm} does not provide an algorithm to reconstruct
a nonseparable signal from its phaseless samples taken on that discrete set.
Considering the phaseless reconstruction of signals in a shift-invariant space, we introduce a local complement property on a set.


\begin{definition}\label{lcp.def0} 
We say that the shift-invariant space $V(\phi)$ has local complement property on a set $A$ if
for any $A'\subset A$, there does not exist  $f, g\in V(\phi)$ such that $f, g\not\equiv 0$ on $A$, but $f({ x})=0$ for all ${x}\in A'$ and $g({y})=0$ for all ${ y}\in A\backslash A'$.
%
 \end{definition}

The  local complement property on the Euclidean  space $\Rd$ is the  complement property  in \cite{YC16}
  for ideal sampling functionals on $V(\phi)$, cf. the complement property for frames in Hilbert/Banach spaces (\cite{alaifari16, BCE06, Bandeira14, cahill15}). Local complement property is closely related to local phase retrievability.
In fact, following the argument in \cite 
{YC16} we have that $V(\phi)$ has the local complement property on $A$ if and only if all signals in $V(\phi)$ is {\em local phase retrievable} on $A$, i.e., for any $ f, g\in V(\phi)$ satisfying $|g({x})|=|f({x})|, { x}\in A$, there exists $\delta\in \{-1, 1\}$ such that $g({ x})=\delta f({x})$ for all ${x}\in A$.  More discussions on the local complement property  is given in Appendix \ref{localcomplementproperty.appendix}.

\begin{theorem}\label{necandsuf.thm}
Let $A_1, \cdots, A_M$ be bounded open sets
and $\phi$ be a compactly supported continuous function such that $\phi$
has local linear independence on $A_m, 1\le m\le M$,
and
\begin{equation}\label{necandsuf.thm.eq1}
S_{k}\cap \big(\cup_{m=1}^M (A_m+\Zd)\big)\ne \emptyset \end{equation}
for all $k\in \Zd$ with  $S_k$ in \eqref{necandsuf.thm.eq2} being nonempty.
 If the shift-invariant space $V(\phi)$ has  local complement property on $A_m, 1\le m\le M$, then there exists  a finite set $\Gamma\subset \cup_{m=1}^M  A_m$ such that the following statements are equivalent for any signal $f\in V(\phi)$:
 \begin{itemize}
 \item[{(i)}] The signal $f$ is determined, up to a sign, from its magnitude measurements on $\Rd$.

  \item[{(ii)}] The graph $\mathcal G_f$ in \eqref{graphG.def} is connected.

  \item[{(iii)}]  The signal $f$ is determined, up to a sign, from its phaseless samples $|f(y)|, y\in \Gamma+\Zd$.
      \end{itemize}
\end{theorem}

The implication (i)$\Longrightarrow$(ii) has been established in Theorem \ref{necessarypr.thm} and the implication (iii)$\Longrightarrow$(i) is obvious.
Write
$f=\sum_{{ k}\in \Zd} c({ k}) \phi(\cdot-{ k})$.
To prove (ii)$\Longrightarrow$(iii), we first  determine
$c({k}), {k}\in {{K}_{A_m}+{l}}$, up to a sign $\epsilon_{m, { l}}\in \{-1, 1\}$, from phaseless samples $|f({\gamma}+{l})|, {\gamma}\in \Gamma$, and then
   we use the connectivity of the graph ${\mathcal G}_f$ to adjust phases $\epsilon_{m, {l}}\in \{-1, 1\},  1\le m\le M, { l}\in \Zd$. Finally  we sew those pieces together to recover amplitudes $c(k), k\in \Zd$, and the signal $f$.
   The detailed argument  will be given in Section \ref{necandsuf.appendix}.
Comparing with the proof of Theorem \ref{tofinite.thm}, we remark that
  our  proof of Theorem \ref{necandsuf.thm}
 is constructive and a reconstruction algorithm can be developed.

\smallskip

 For the case that  the generator $\phi$ has local linear independence on any open set, we can find open sets $A_m, 1\le m\le M$, such that
 \eqref{necandsuf.thm.eq1} holds and $V(\phi)$ has local complement property on $A_m, 1\le m\le M$, see Proposition \ref{llilocalcomplement.pr}. Then from Theorem \ref{necandsuf.thm} we obtain the following corollary, cf. Theorem  
\ref{sufficient.thm} and Corollaries \ref{wavelet.cor} and \ref{spline222.cor}.
 \begin{corollary} Let $\phi$ be a compactly supported continuous function such that $\phi$
has local linear independence on any open set.  Then there exists
a finite set $\Gamma$ such that any nonseparable
 signal  
is determined, up to a sign, from its phaseless samples
taken on the set $\Gamma+\Zd$ with finite sampling density.
 \end{corollary}

Take  ${\bf N}=(N_1, \ldots, N_d)^{T}$ with $N_i\ge 2, 1\le i\le d$, and
let $B_{N_i}$ be the B-spline of order $N_i$ (\cite{deboorbook, unser99, wahba90}). Define the box spline function of tensor-product  type 
\begin{equation}\label{boxspline2.def}
B_{\bf N}({x}):= B_{N_1}(x_1)\times\cdots\times B_{N_d}(x_d), \ { x}=(x_1,\ldots, x_d)^T\in \R^d.
\end{equation}
As the restriction of a signal in $V(B_{\bf N})$ on $(0, 1)^d$
is a polynomial of finite degree, the space $V(B_{\bf N})$ has the local complement property on $(0, 1)^d$. Applying Theorem \ref{necandsuf.thm}  with $M=1$ and $A_1=(0, 1)^d$
 leads to the following result for
 tensor-product splines, which is given in \cite{YC16} for $d=1$.

\begin{corollary} \label{tensor.cor}  
Let $X_i$ contain $2N_i-1$ distinct points in $(0, 1), 1\le i\le d$.
Then any nonseparable signal $f\in V(B_{\bf N})$  can be reconstructed, up to a sign, from its phaseless samples on  the set
$X_1\times \ldots\times X_d+\Zd$
with sampling  density $\prod_{i=1}^d (2N_i-1)$.
\end{corollary}

The detailed proof of the above corollary is given in Section \ref{tensor.proof}.

\smallskip

In  the proof of Theorem \ref{necandsuf.thm}  given in Section \ref{necandsuf.appendix},
the discrete sampling set $\Gamma$ is chosen to be the union of  $\Gamma_m\subset A_m, 1 \le m\le M$,
\begin{equation}\label{gamma.1} \Gamma=\cup_{m=1}^M \Gamma_m,
\end{equation}
so that
outer products
$\Phi_{A_m}({\gamma}) (\Phi_{A_m}({\gamma}))^T, \gamma\in \Gamma_m$,
is a basis (or
 a spanning set) of the linear space  $W_{ A_m}$. 
Therefore we have the following result from  Theorem \ref{necandsuf.thm}.

 \begin{corollary} \label{necandsuf.cor}
 Let  $\phi$ and $A_m, 1\le m\le M$, be as in Theorem \ref{necandsuf.thm}. Then
 any nonseparable signal $f\in V(\phi)$  can be reconstructed from its phaseless samples on
a shift-invariant set $\Gamma+\Zd$ with sampling density
$$D(\Gamma+\Zd)\le \sum_{m=1}^M  {\rm dim} W_{A_m}\le \frac{1}{2}\sum_{m=1}^M \# K_{A_m} (\#K_{A_m}+1).$$
\end{corollary}

The discrete set  $\Gamma+\Zd$ chosen in  Corollary \ref{necandsuf.cor}
may have larger sampling density than ${\rm dim} \ W_{(0,1)^d}$ in Corollary \ref{corollary32}. 
Based on the constructive proof in Theorem \ref{necandsuf.thm},   a robust reconstruction algorithm from  phaseless samples taken on the discrete set
is developed in Section  \ref{numerical.sec}.
However, we have difficulty to find a  reconstruction algorithm from the phaseless samples taken on  the set given in Corollary \ref{corollary32}.

\begin{definition} {\rm We say that ${\mathcal M}=\{a_m\in \R^d, 1\le m\le M\}$ is
 a {\em phase retrievable frame} for $\R^d$ if
any vector ${x}\in \R^d$ is determined, up to a sign, from its measurements
$|\langle {x}, a_m\rangle|, a_m\in {\mathcal M}$, and a
{\em minimal phase retrieval frame} for $\R^d$
if any true subset of ${\mathcal M}$ is not a phase retrievable frame.
}
\end{definition}

 After careful examination on the proof of Theorem \ref{necandsuf.thm}, we can
select a subset $\Gamma'$ of  $\Gamma$ such that all nonseparable signals $f$ can be reconstructed from its phaseless samples taken on $\Gamma' + \Z^d$ in a robust manner.

 \begin{theorem}\label{finite.cor}  Let  $A_m, 1\le m\le M$,
  and $\phi$ be as in Theorem \ref{necandsuf.thm}.
Assume that  there exist $\Gamma_m^\prime\subset A_m$ such that
$\Phi_{A_m}(\gamma), \gamma\in \Gamma_m^\prime$, is a phase retrievable frame for $\R^{\# K_{A_m}}$.
Then
 any nonseparable signal $f\in V(\phi)$ is determined, up to a sign, from its phaseless samples $|f(y)|, y \in \Gamma'+\Zd$, where
 \begin{equation}\label{gamma.2}
 \Gamma'=\cup_{m=1}^M \Gamma_m^\prime.
 \end{equation}
      \end{theorem}

In  Theorem \ref{finite.cor}, the requirement on the sampling set is a bit weaker than the one in  Theorem \ref{necandsuf.thm},
as  for the sampling set $\Gamma=\cup_{m=1}^M \Gamma_m$ in \eqref{gamma.1},
$\Phi_{A_m}({\gamma}),  \gamma\in \Gamma_m$, is  a phase retrievable frame for $\R^{\# K_{A_m}}$, cf. Theorem \ref{localcomplementcriterion.thm}.
We remark that the
 phase retrieval frame property for
 $\Phi_{A}({\gamma}), {\gamma}\in \Gamma$, may not imply that
 their out products
$\Phi_{A}({\gamma}) (\Phi_{A}({\gamma}))^T, \gamma\in  \Gamma$,
form a basis (or a spanning set) of  $W_{A}$ in \eqref{WPhiA.def},  as shown in the following example.

\begin{example} {\rm Let
\begin{equation*}
\phi_1(x)=\left\{\begin{array}{ll} x^3/2 & {\rm if} \ 0\le x<1\\
-x^3+3x^2-2x+1/2 & {\rm if} \ 1\le x<2\\
x^3/2-3x^2+5x-3/2 & {\rm if} \ 2\le x<3\\
0 & {\rm otherwise},
\end{array}
\right.
\end{equation*}
and set $\Phi_1(x)= (\phi_1(x), \phi_1(x+1), \phi_1(x+2))^T, 0\le x<1$. Then
$$\Phi_1(x)
=\frac{1}{2}
\begin{pmatrix} 0\\
1\\
1\end{pmatrix}+
\begin{pmatrix} 0\\
1\\
-1\end{pmatrix} x+ \frac{1}{2}\begin{pmatrix} 1\\
-2\\
1\end{pmatrix} x^3,  
$$
and
\begin{eqnarray*}
\hskip-0.3in & & \Phi_1(x) \Phi_1(x)^T  =   \frac{1}{4}
\begin{pmatrix} 0 & 0 & 0\\
0 & 1 & 1\\
0 & 1 & 1\end{pmatrix}+  \begin{pmatrix} 0 & 0 & 0\\
0 & 1 & 0\\
0 & 0 & -1\end{pmatrix} x
+\begin{pmatrix} 0 & 0 & 0\\
0 & 1 & -1\\
0 & -1 & 1\end{pmatrix} x^2\nonumber\\
\hskip-0.2in & & \quad +  \frac{1}{4} \begin{pmatrix} 0 & 1 & 1\\
1 & -4 & -1\\
1 & -1 & 2\end{pmatrix} x^3 + \frac{1}{2} \begin{pmatrix} 0 & 1 & -1\\
1 & -4 & 3\\
-1 & 3 & -2\end{pmatrix} x^4+ \frac{1}{4}\begin{pmatrix} 1 & -2 & 1\\
-2 & 4 & -2\\
1& -2 & 1\end{pmatrix}
 x^6.
\end{eqnarray*}
Therefore the space spanned by $\Phi_1(x), 0< x<1$, is  $\R^3$,
and  the space $W_{(0,1)}$ spanned by $\Phi_1(x) \Phi_1(x)^T, 0< x<1$, is the $6$-dimensional linear space of
symmetric matrices of size $3\times 3$.
On the other hand, any $3\times 3$ square submatrices of  
$$ \Big(\Phi_1(0)\  \Phi_1\big(\frac{1}{5}\big)\  \Phi_1\big(\frac{2}{5}\big)\  \Phi_1\big(\frac{3}{5}\big)\  \Phi_1\big(\frac{4}{5}\big)\Big)
= \frac{1}{250}
\begin{pmatrix}          0    &  1  &     8  &    27  &   64\\
  125  &  173  &   209  &   221  &  197\\
    125 &    76 &    33 &     2 &   -11 \end{pmatrix}
    $$
    is  nonsingular, 
which implies that
$\Phi_1(m/5), 0\le m\le 4$,
forms a phase retrieval frame for $\R^3$, but their out products do not form a spanning set of the $6$-dimensional space $W_{(0,1)}$.
}\end{example}

The problem how to pare down  a phase retrieval frame to a minimal phase
retrieval frame
 will be discussed in our future work. Using the pare-down technique, we may find a discrete set  $X$ with smaller sampling density such that nonseparable signals
 in the shift-invariant space can be reconstructed from their phaseless samples taken on $X$.

\section{Stability of phaseless sampling and reconstruction}
\label{reconstruction.section}

Stability  is of paramount importance in the phaseless sampling and reconstruction problem.
Consider the scenario that
phaseless samples of a signal  
\begin{equation}\label{signal.def2}
f=\sum_{k\in \Zd} c(k) \phi(\cdot-k)\in V(\phi)
\end{equation}
taken on a shift-invariant set $\Gamma+\Zd$ are corrupted
by the additive noise,
\begin{equation}\label{noisydata.def}
{ z}_{\epsilon}({y})= |f({y})|+\epsilon({ y}), \ {y}\in \Gamma+\Zd,
\end{equation}
 where   $ \epsilon =(\epsilon({y}))_{y\in \Gamma+\Zd}$
 has the bounded noise level
$\|\epsilon\|_{\infty}=\max_{y\in \Gamma+\Zd} |\epsilon(y)|$,
and
$\Gamma=\cup_{m=1}^M \Gamma_m$ is either as in \eqref{gamma.1} or
 in \eqref{gamma.2}.  In this section, we construct
 an 
 approximation
 \begin{equation}\label{fepsilon.def}
 f_{\epsilon}=\sum_{k\in \Zd} c_{ \epsilon}(k) \phi(\cdot-k)\in V(\phi),\end{equation}
 up to a sign, to the original signal $f$ in \eqref{signal.def2}  when  the noisy phaseless samples  \eqref{noisydata.def} are available only.

Let
\begin{equation}\label{omegaset.def}
\Omega_{m}=\{k\in \Zd:\ \phi({\gamma-k})\ne 0 \ {\rm for  \ some } \ {\gamma}\in \Gamma_m\}, \ 1\le m\le M,
\end{equation}
and define the hard threshold function $H_\eta, \eta\ge 0$,  by
 $$H_\eta(t)=\left\{ \begin{array}{ll} t  & {\rm if} \ |t|\ge \eta\\
 0  & {\rm if}\ |t|<\eta.
 \end{array}\right. $$
Based on the constructive proofs of Theorems \ref{necandsuf.thm} and  
\ref{finite.cor}, we propose the following four-step approach  with its implementation discussed in Section \ref{numerical.sec}.
\begin{tcolorbox}[colback=white,colframe=black]
 \begin{itemize}
 \item[{1.}] Select a phase adjustment threshold $M_0\ge 0$ and an amplitude threshold $\eta=\sqrt{M_0}$.
\item[{2.}]
For $l\in \Zd$ and $1\le m\le M$, let \begin{equation}\label{meps.def1}
{c}_{\epsilon,  l; m}=(c_{\epsilon,  l; m}(k))_{k\in \Zd}
\end{equation}
take zero components except that
 $c_{\epsilon, l;  m}(k), k\in l+\Omega_{m}$, are  solutions of the  local  minimization problem
{
 \begin{eqnarray} \label{meps.def2}
& & \min_{c(k), k\in l+\Omega_m}  \sum_{{ \gamma} \in \Gamma_m} \Bigg|\Big|\sum_{k\in l + \Omega_{ m}} c(k) \phi({\gamma}+{l}-k)\Big|-z_{\epsilon}({ \gamma} +{l} )\Bigg|^2.
\end{eqnarray}}

\item [{3.}] Adjust  phases of ${c}_{\epsilon, {l}; m}$ appropriately
so that
the resulting vectors $ \delta_{l, m}{ c}_{\epsilon,l; m}$
with $\delta_{l, m}\in \{-1, 1\}$ satisfy
\begin{equation} \label{meps.def13}
 \langle   \delta_{l, m} { c}_{\epsilon, l;  m},  \delta_{l', m'} { c}_{\epsilon, l'; m'}\rangle\ge -M_0
\end{equation}
 for   all $l, l'\in \Zd$ and $1\le m, m'\le M$.
 \item[{4.}] Sew  vectors $\delta_{l, m}{c}_{\epsilon, l; m}, l\in \Zd, 1\le m\le M$, together to obtain
\begin{equation}\label{meps.def14}
d_{\epsilon}(k)=\frac{\sum_{m=1}^M\sum_{l\in \Zd}  \delta_{l, m}{ c}_{\epsilon, l; m}(k)}{\sum_{m=1}^M\sum_{ l\in \Zd} \chi_{l+\Omega_{m}}(k)}, \ k\in \Zd.
\end{equation}
\item [{5.}] Threshold the vector  $d_\epsilon=(d_\epsilon(k))_{k\in \Zd}$,
\begin{equation}\label{thresholding}
c_{\epsilon}(k)= H_{\eta} (d_\epsilon(k)), \ k\in \Zd
\end{equation}
to define the approximation $ f_{\epsilon}$ in \eqref{fepsilon.def}.
 \end{itemize}
\end{tcolorbox}

In the next theorem, we show that the above approach provides a suboptimal approximation to the original signal in a noisy phaseless sampling environment.

\begin{theorem}\label{stability.thm}
Let $A_1, \cdots, A_M$ be bounded open sets satisfying
\eqref{necandsuf.thm.eq1},
$\phi$ be a compactly supported continuous function such that $\phi$
has local linear independence on $A_m, 1\le m\le M$,
and $\Gamma_m\subset A_m$  be so chosen that
$\Phi_{A_m}(\gamma), \gamma\in \Gamma_m$, is a phase retrievable frame for $\R^{K_{A_m}}$. Assume that  the graph ${\mathcal G}_f=(V_f, E_f)$ of the original signal
 $f=\sum_{k\in \Z^d} c(k) \phi(\cdot-k)$ is connected and
\begin{equation}\label{stability.thm.eq0}
F_0:=\inf_{k\in  V_f}  |c(k)|^2>0.
\end{equation}
Set  $\Gamma=\cup_{m=1}^M \Gamma_m$ and
\begin{eqnarray} \label{phil2.def}
\hskip-0.05in&\hskip-0.05in \hskip-0.05in& \hskip-0.05in\big\|\Phi^{-1}\big\|_2=\sup_{\Theta_m\subset \Gamma_m, 1\le m\le M}
 \Big(\min\big(
\sup_{ \|d\|_2=1} \|\Phi_{\Theta_m} d\|_2^{-1},\nonumber\\
& & \hskip2in \qquad \qquad  \sup_{ \|d\|_2=1} \|\Phi_{\Gamma_m\backslash\Theta_m} d\|_2^{-1}\big)\Big)^{-1},  
\end{eqnarray}
where $\Phi_{\Theta_m}=(\phi(\gamma-k))_{\gamma\in \Theta_m, k\in \Omega_m}$
for $\Theta_m\subset \Gamma_m$.
If
 the phase adjustment threshold constant $M_0\ge 0$ and the noise level
 $\|\epsilon\|_\infty:=\sup_{{y}\in \Gamma+\Zd} |\epsilon(y)|$
 satisfy
\begin{equation}  \label{stability.thm.eq1-}
  M_0\le \frac{2 F_0}{9},
\end{equation}
and
\begin{equation}  \label{stability.thm.eq1}
8   \# \Gamma  \|\Phi^{-1}\|_2^2  \|\epsilon\|_\infty^2\le  M_0, 
\end{equation}
then
the  signal
 $f_{\epsilon}=\sum_{k\in\Zd}c_{ \epsilon}(k)\phi(\cdot-k)\in V(\phi)$
 reconstructed from the proposed approach \eqref{meps.def1}--\eqref{thresholding} 
  satisfies
\begin{equation}\label{stability.thm.eq3}
|c_{\epsilon}(k)-\delta c(k)|\le 2 \sqrt{\# \Gamma} \|\Phi^{-1}\|_2 \|\epsilon\|_\infty, \ k\in V_f
\end{equation}
and
\begin{equation}\label{stability.thm.eq3+}
c_{\epsilon}(k)=c(k)=0, \ k\not\in V_f,
\end{equation}
where   $\delta\in\{-1,1\}$.
 \end{theorem}

By Theorem \ref{stability.thm}, the  reconstructed signal $f_\epsilon$ in \eqref{fepsilon.def} provides
a suboptimal approximation, up to a sign,  to the original signal $f$ in \eqref{signal.def2},
\begin{equation}\label{noresonance.def}
 \|f_\epsilon- \delta f\|_\infty
\le  2 \sqrt{\# \Gamma} \|\Phi^{-1}\|_2 \Big(\sup_{x\in \Rd} \sum_{k\in \Zd} |\phi(x-k)|\Big) \|\epsilon \|_\infty
\end{equation}
and
\begin{equation}\label{separable_error.def}\sup_{y\in \Gamma+\Zd}
\big||f_\epsilon(y)|- |f(y)|\big|
\le  2 \sqrt{\# \Gamma} \|\Phi^{-1}\|_2 \Big(\sup_{x\in \Rd} \sum_{k\in \Zd} |\phi(x-k)|\Big) \|\epsilon \|_\infty.
\end{equation}

By \eqref{stability.thm.eq0}, \eqref {stability.thm.eq1-}, \eqref{stability.thm.eq1} and \eqref{stability.thm.eq3}, a vertex in  the graph ${\mathcal G}_f$
 is also a  vertex  of the  graph ${\mathcal G}_{f_\epsilon}$ associated with
 the reconstructed signal $f_\epsilon$. This together with  \eqref{graphG.def} and  \eqref{stability.thm.eq3+} implies that the graphs
 ${\mathcal G}_f$  and  ${\mathcal G}_{f_\epsilon}$ associated with the original signal
 $f$ and the reconstructed signal $f_\epsilon$ are the same, i.e.,
 $${\mathcal G}_f={\mathcal G}_{f_\epsilon}.$$

\smallskip

The selection of the threshold constant $M_0\ge 0$ is imperative
 to find an approximation to the original signal from its phaseless samples.
   In the noiseless environment, we may take $M_0=0$ and
   the proposed  approach leads to a perfect reconstruction, i.e., $f_\epsilon=\pm f$, when $f$ is nonseparable.
In  practical applications,  the noise level is positive and  the phase adjustment threshold constant $M_0$ needs to be appropriately selected. For instance, we may require that \eqref{stability.thm.eq1-} and \eqref{stability.thm.eq1} are satisfied if we have some prior information about the amplitude vector of the original signal.
From the proof of Theorem \ref{stability.thm} and also the simulations in the next section, it is observed that phases can not be adjusted to satisfy  \eqref{meps.def13} if $M_0$ is far below square of noise level $\|\epsilon\|_\infty$ (for instance, \eqref{stability.thm.eq1} is not satisfied), while the phase adjustment  \eqref{meps.def13} in the algorithm is not essentially determined   and hence the reconstructed signal is not a good approximation of the original signal if $M_0$ is comparable to the square of
minimal magnitude of amplitude vector of the original signal (for instance, \eqref{stability.thm.eq1-} is not satisfied).

\smallskip

\begin{remark}\label{remark5.2}
{\em By Theorem \ref{stability.thm},  there is   no resonance phenomenon
in the sense that
\begin{equation}\label{resonance.def23}
\inf_{\delta\in \{-1, 1\}}\|f_\epsilon-\delta f\|_\infty \le C \|\epsilon\|_\infty
\end{equation}
if  the noise level is far below  the minimal magnitude of amplitude vector of the original signal, i.e.,
\begin{equation}\label{epsilon.req}
\|\epsilon\|_\infty \le  C_0 \inf_{k\in V_f} |c(k)|\end{equation}
for some  sufficiently small constant $C_0$.
The phaseless sampling and reconstruction problem is ill-posed if the noise level is high. For instance, the estimate
\eqref{resonance.def23} is not satisfied for the nonseparable spline signal of order $2$,
$$f_\alpha(x)= B_2(x)+\alpha B_2(x-1)+B_2(x-2)\in V(B_2),$$
 if $\|\epsilon\|_\infty \ge 2\alpha/(1+\alpha)$, where $\alpha\in (0,1)$
is sufficiently small. The reasons are that the  signal
$\tilde f_\alpha(x)= B_2(x)+\alpha B_2(x-1)-B_2(x-2)\in V(B_2)$
satisfy
$$\min_{\delta\in \{-1, 1\}} \|f_\alpha-\delta \tilde f_\alpha\|_{\infty}=2\ \  {\rm  and}\  \ \big\||f_\alpha|-|\tilde f_\alpha |\big\|_\infty=\frac{2\alpha}{1+\alpha}.$$
}
\end{remark}


\section{Reconstruction Algorithm and Numerical Simulations}
\label{numerical.sec}

Consider the scenario that
phaseless samples of a signal  
$f=\sum_{k\in \Zd} c(k) \phi(\cdot-k)\in V(\phi)$
taken on a finite 
 set
$\Gamma+ K\subset \Gamma+\Zd$
are corrupted
by the additive noise,
\begin{equation}\label{noisydata2.def}
{ z}_{\epsilon}({y})= |f({y})|+\epsilon({ y}), \ {y}\in \Gamma+K,
\end{equation}
 where   $ \epsilon(y)\in [-\epsilon, \epsilon] , y\in \Gamma+K$,
for some $\epsilon \ge 0$,
and
$\Gamma=\cup_{m=1}^M \Gamma_m$ is either as in \eqref{gamma.1} or
 in \eqref{gamma.2}.
 Define
  \begin{equation}\label{original}
 f_K=\sum_{k\in \widetilde K} c(k) \phi(\cdot-k),
  \end{equation}
 where  
$ \widetilde K=\cup_{l\in K}\cup_{m=1}^M (l+\Omega_m)$  
 and
   $\Omega_m, 1\le m\le M$, are as in \eqref{omegaset.def}.
Then the noisy data $z_{\epsilon}(y), y\in \Gamma+K$, in \eqref{noisydata2.def} is
\begin{equation}\label{noisydata3.def}
{z}_{\epsilon}({y}) =|f_K({y})|+\epsilon({ y})\ge 0, \ {y}\in \Gamma+K.\end{equation}
 Based on \eqref{noisydata3.def} and the four-step approach in Section \ref{reconstruction.section}, we propose an algorithm
to find
an approximation $f_\epsilon$ of the form
 \begin{equation}\label{fepsilon2.def}
 f_{\epsilon}=\sum_{k\in \widetilde K} c_{ \epsilon}(k) \phi(\cdot-k)\in V(\phi),\end{equation}
 up to a sign, to the original signal $f$ in \eqref{original}  when  the noisy phaseless samples  \eqref{noisydata2.def}  are available only.
    The algorithm contains four parts: {\bf m}inimiztion, {\bf a}djusting {\bf p}hases, {\bf se}wing and {\bf t}hresholding, and we call it  the MAPSET algorithm.
\begin{algorithm}\label{mapset.al}
\caption{MAPSET Algorithm}
\begin{algorithmic}
\STATE {\bf Input}: finite set $K\subset \Zd$;
sampling set $\Gamma=\cup_{m=1}^M \Gamma_m$  either  in \eqref{gamma.1} or
 in \eqref{gamma.2};  noisy phaseless sampling data $\big(z_{\epsilon}(y)\big)_{y\in \Gamma+ K}$;
  index set
 $\widetilde K= \cup_{l\in K}\cup_{m=1}^M (l+\Omega_m)\subset \Zd$; and the phase adjustment threshold constant
 $M_0$.

 \STATE {\bf Initials}: Start from zero vectors
$c_{\epsilon, l; m}=\big(c_{\epsilon, l; m}(k)\big)_{k\in \widetilde K}, l\in K, 1\le m\le M$.

\STATE {\bf Instructions}:

\STATE {\bf 1)  Local minimization}:
For $l\in { K}$ and $1\le m\le M$, replace $c_{\epsilon, l; m}(k), k\in l+\Omega_m$, by a solution
of the  
 local minimization problem
{\small \begin{eqnarray*} \label{meps.def2new}
& & \min_{c(k), k\in l+\Omega_m}  \sum_{{ \gamma} \in \Gamma_m} \Bigg|\Big|\sum_{k\in l + \Omega_{ m}} c(k) \phi({\gamma}+{l}-k)\Big|-z_{\epsilon}({ \gamma} +{l} )\Bigg|^2.
\end{eqnarray*}}

\STATE {\bf 2)  Phase adjustment}:
For $l\in K$ and $1\le m\le M$,  multiply  $c_{\epsilon, l; m}$ by $\delta_{l, m}\in \{-1, 1\}$ so that
$\langle \delta_{l,m}c_{\epsilon, l; m}, \delta_{l', m'} c_{\epsilon, l'; m'}\rangle \ge -M_0
$
for all $l, l'\in K$ and $1\le m, m'\le M$.
\STATE {\bf 3) Sewing local approximations:}
\begin{equation*}
d_{\epsilon}(k)=\frac{\sum_{m=1}^M\sum_{l\in K}  \delta_{l, m}c_{\epsilon, l; m}(k)}{\sum_{m=1}^M\sum_{ l\in K} \chi_{l+\Omega_{m}}(k)}, \ k\in \widetilde K.
\end{equation*}
\STATE {\bf 4)  Hard thresholding:}
\begin{equation*}
c_{\epsilon}(k)= \left\{ \begin{array}{ll} d_{\epsilon}(k)  & {\rm if} \ |d_{\epsilon}(k) |\ge \sqrt{M_0}\\
 0  & {\rm else}
 \end{array}\right.,  \ k\in \widetilde K.
\end{equation*}
\STATE {\bf Output}: Amplitude vector $(c_{\epsilon}(k))_{k\in \widetilde K}$, and the
 reconstructed signal $f_\epsilon=\sum_{k\in \widetilde K} c_\epsilon(k) \phi(\cdot-k)$.
\end{algorithmic}
\end{algorithm}

In this section, we also demonstrate the performance of the proposed MAPSET algorithm on reconstructing  box spline signals from their  noisy phaseless samples on a discrete set.

\subsection{Nonseparable spline signals of
tensor-product type}\label{tensor.al}
Let $B_{(3,3)}$ be the tensor product of  one-dimensional quadratic spline  $B_3$,
see \eqref{boxspline2.def}.
For $A=(0, 1)^2$ and $\phi=B_{(3, 3)}$, the vector-valued function $\Phi_A$ in \eqref{phiA.def}  and the set $K_A$ in \eqref{KA.def}
becomes
 \begin{equation}\label{phiAtensor.def} \Phi_{(0,1)^2}(s, t)=\big(b_i(s) b_j(t)\big)_{(i,j)\in K_{(0,1)^2}},\ (s,  t)\in (0,1)^2\end{equation}
and
$K_{(0,1)^2}=\{(i,j): -2\le i, j\le 0\}$
respectively, where $b_0(s)= s^2/2$,  $b_{-1}(s)=(-2s^2+2s+1)/2$ and $b_{-2}(s)= (1-s)^2/2,  0\le s\le 1$.
One may verify that
 the space spanned by the outer products of $\Phi_{(0,1)^2}(s, t), (s, t)\in (0, 1)^2$,
has dimension $25$, and the set
\begin{equation}\label{gamma00}
\Gamma_0=\{(i,j)/6, 1\le i,j\le 5\} \subset (0, 1)^2\end{equation}
 with cardinality $25$   satisfies \eqref{gamma.1}, see  Figure \ref{sample.pic}.
  For the above uniformly distributed set $\Gamma_0$,
  the corresponding $\|\Phi^{-1}\|_2$ in  \eqref{phil2.def}  is
 $2.7962\times10^{3}$.  

  As $\Phi_{(0,1)^2}(s, t), (s, t)\in (0, 1)^2$,
  is a  $9$-dimensional vector-valued polynomial about $s^mt^n, 0\le m, n\le 2$,
  the shift-invariant space generated by $B_{(3, 3)}$ has local complement property on $(0, 1)^2$.   Observe that the matrix
  $\big(\Phi_{(0,1)^2}(s_i, t_i)\big)_{1\le i\le 9}$ has full rank $9$ for almost all
  $(s_i, t_i)\in (0, 1)^2, 1\le i\le 9$. Hence
    $\big(\Phi_{(0,1)^2}(s_i, t_i)\big)_{1\le i\le 17}$ is a phase retrieval frame for almost all
  $(s_i, t_i)\in (0, 1)^2, 1\le i\le 17$,  but
 the corresponding $\|\Phi^{-1}\|_2$ in  \eqref{phil2.def} are relatively large
 from  our numerical calculation.  So we 
use a randomly distributed set
 $\Gamma_1\subset (0,1)^2$ with cardinality 19 in our simulations, see Figure \ref{sample.pic}. The above set satisfies \eqref{gamma.2} and the corresponding $\|\Phi^{-1}\|_2$ in \eqref{phil2.def} is
$3.2995\times10^{4}$.  
  \begin{figure}[H]  
\begin{center}
\includegraphics[width=58mm, height=42mm]{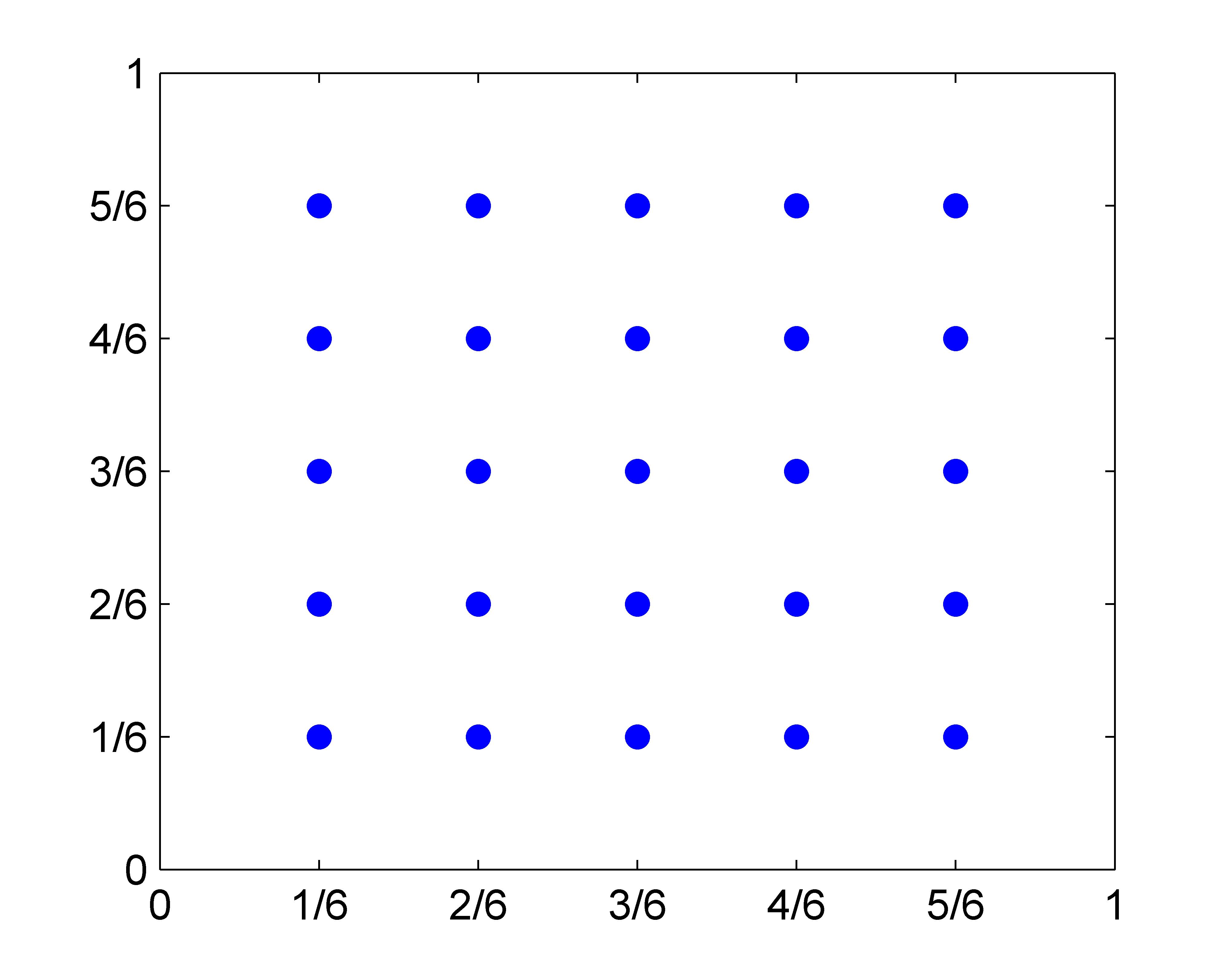}
\includegraphics[width=58mm, height=42mm]{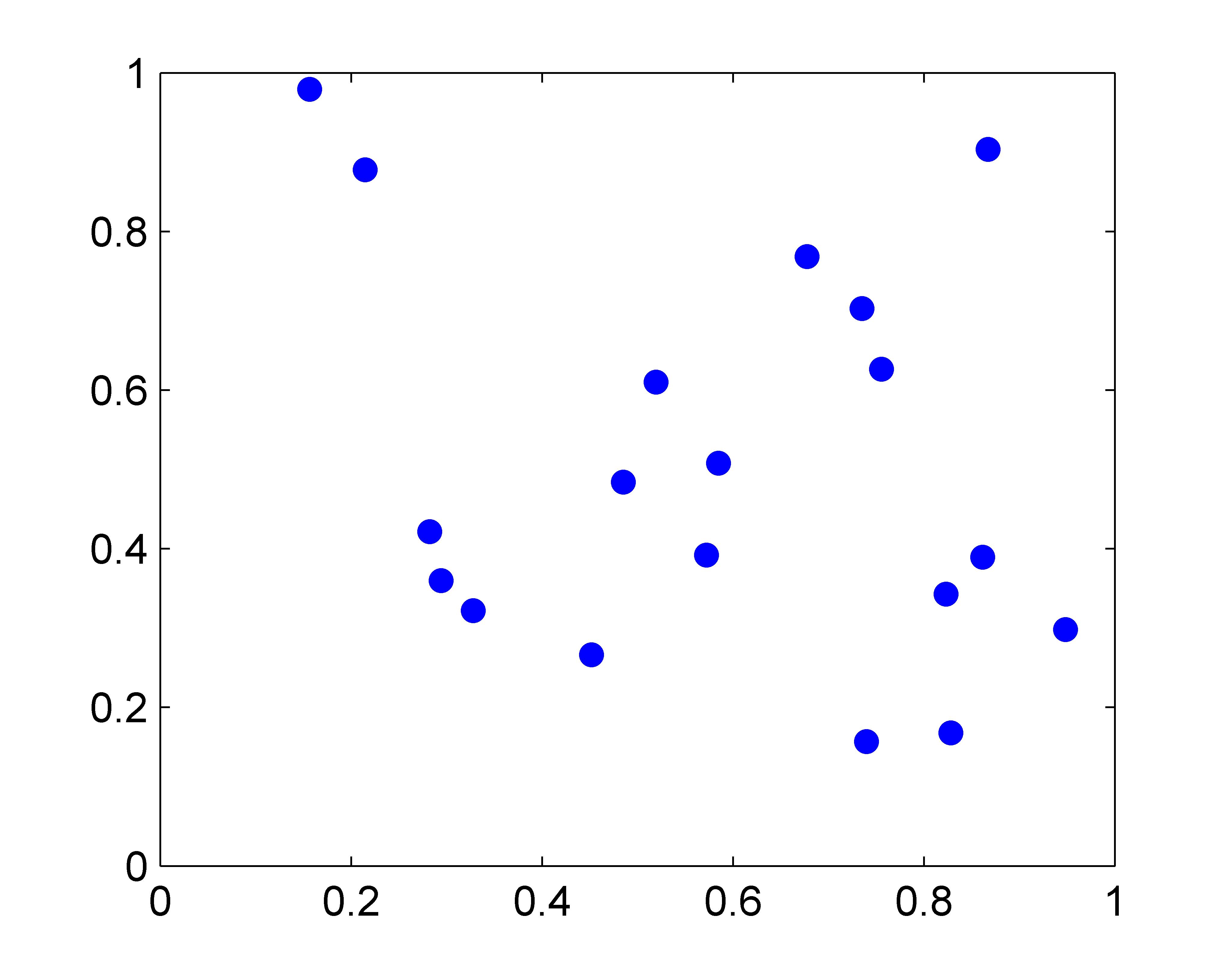}
\caption{Plotted on the left is a uniformly distributed set  $\Gamma_0$
 satisfying \eqref{gamma.1}, while on the right is
a randomly distributed set  $\Gamma_1$  satisfying \eqref{gamma.2}.
 The corresponding $\|\Phi^{-1}\|_2$ in \eqref{phil2.def} to the above sets are
 $2.7962\times10^{3}$ (left) and $3.2995\times10^{4}$ (right), respectively.
  }
\label{sample.pic}
\end{center}
\end{figure}

In our simulations, the available  data  $z_\epsilon(y)=|f(y)|+\epsilon(y)\ge 0, y\in \Gamma+K$,
  are  noisy phaseless samples of a  spline signal
  \begin{equation}\label{tensorsignal.def0}
  f(s, t)=\sum_{0\le m\le K_1, 0\le  n\le K_2} c(m,n)B_{(3,3)}(s-m, t-n),  \end{equation}
taken on $\Gamma+K$, where
 $K=[0, K_1]\times [0, K_2]$ for some positive integers $K_1, K_2\ge 1$,
 $\Gamma$ is either the uniform set $\Gamma_0$ or the random set $\Gamma_1$  in
Figure \ref{sample.pic},
 amplitudes of the signal $f$,
\begin{equation}\label{amplitude.def}
c(m,n)\in[-1, 1]\backslash [-0.1, 0.1] ,\   0\le m\le K_1, 0\le n\le K_2, \end{equation} are randomly chosen, and   the additive noises $\epsilon(y)\in [-\epsilon, \epsilon], y\in \Gamma +K$,
  with noise level $\epsilon\ge 0$ are randomly selected.
Denote the signal reconstructed by the proposed MAPSET algorithm with phase adjustment thresholding constant $M_0=0.01$, cf. \eqref{stability.thm.eq1-} with $F_0=0.01$,  by
 \begin{equation}
 f_{ \epsilon}(s, t)=\sum_{-2\le m\le K_1, -2\le n\le K_2}c_{\epsilon}(m,n)B_{(3,3)}(s-m, t-n).
 \end{equation}

 Define maximal amplitude error of the MAPSET algorithm  by
   \begin{equation}\label{rec.error.def}
 e(\epsilon):=\min_{\delta\in \{-1,1\}}\max_{-2\le m\le K_1, -2\le n\le K_2}|c_{\epsilon}(m,  n)-\delta c(m,n)|.
 \end{equation}
 As the original spline signal $f$ in \eqref{tensorsignal.def0} is nonseparable, the conclusions
\eqref{stability.thm.eq3} and \eqref{stability.thm.eq3+}
guarantee that the
reconstruction signal $f_\epsilon$  provides an approximation, up to a sign, to the original signal $f$ if  $\|\Phi^{-1}\|_{2}\epsilon$ is much smaller than a multiple of $\sqrt{M_0}$, where $M_0$ is the  phase adjustment thresholding constant.
 Our numerical simulation indicates that
  the MAPSET algorithm saves phases successfully
  in $100$ trials and the maximal amplitude error  $e(\epsilon)$ in \eqref{rec.error.def} is about
   $O(\epsilon)$,
  provided that  $\epsilon \le 2\times 10^{-3}$ for  $\Gamma=\Gamma_0$ and $\epsilon \le  7\times 10^{-4}$ for  $\Gamma=\Gamma_1$,
  where $\sqrt{M_0}/\|\Phi^{-1}\|_{2}$ are $3.5763\times 10^{-5}$ and  $3.0307 \times 10^{-6}$ respectively.  Presented in Figure \ref{signalandrec.fig} are
  a spline  signal  $f$ in  \eqref{tensorsignal.def0} with $K_1= K_2 = 9$ and the difference between the original signal $f$ and
 the reconstructed signal $f_{\epsilon}$ via the MAPSET algorithm with  noise level $\epsilon =10^{-4}$.
\begin{figure}[h] 
\begin{center}
\includegraphics[width=40mm, height=28mm]{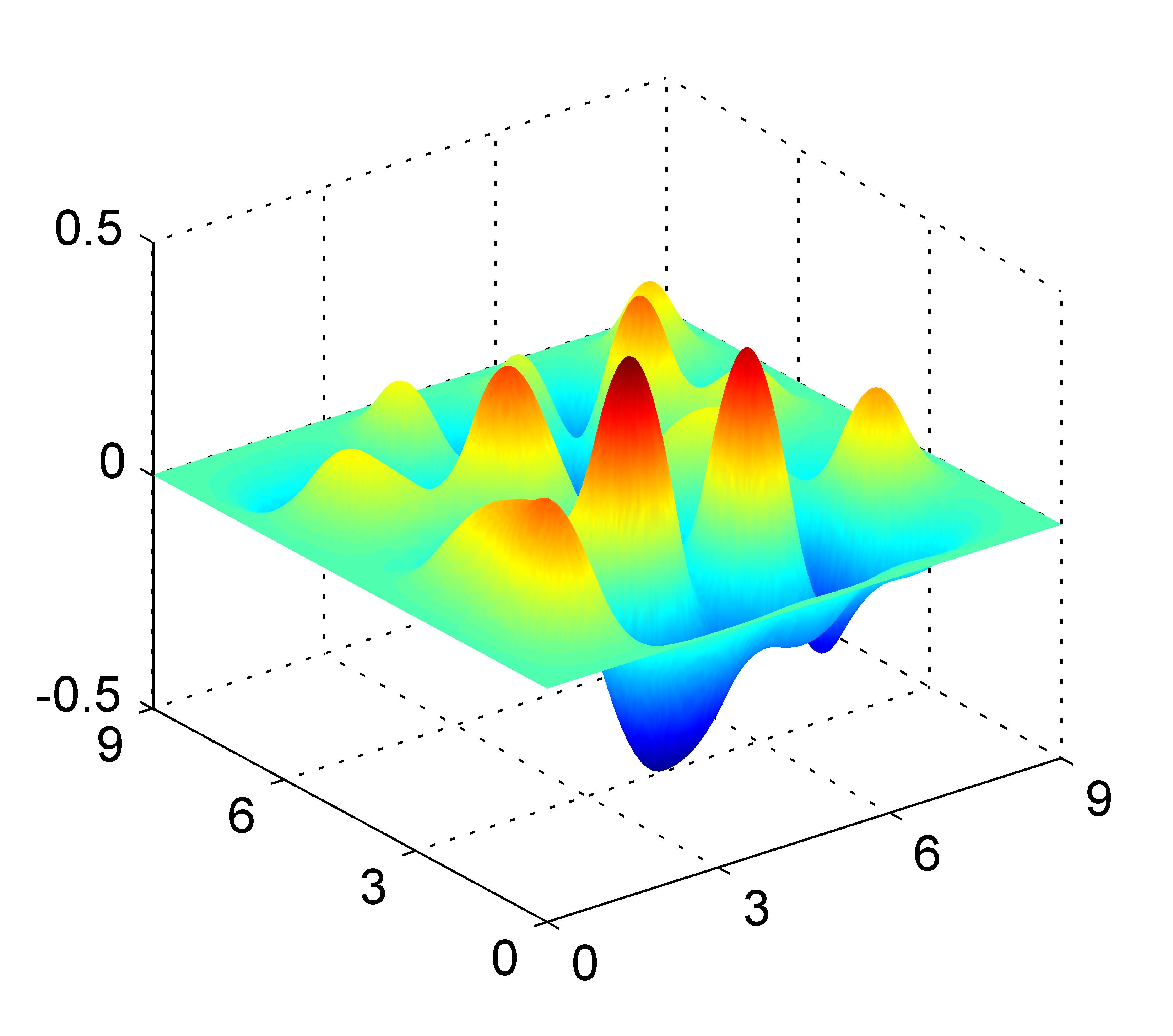}
\includegraphics[width=40mm, height=28mm]{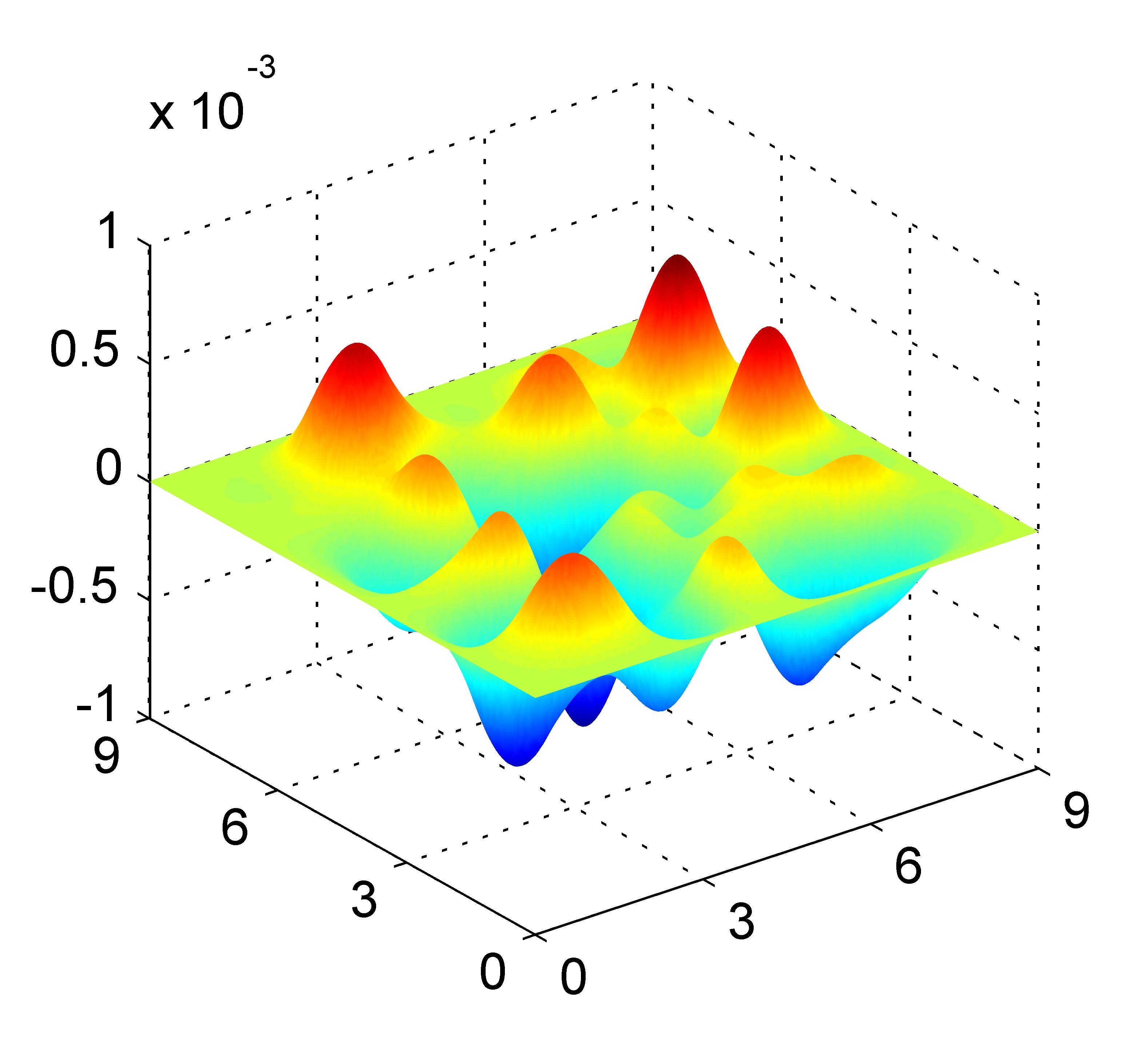}
\includegraphics[width=40mm, height=28mm]{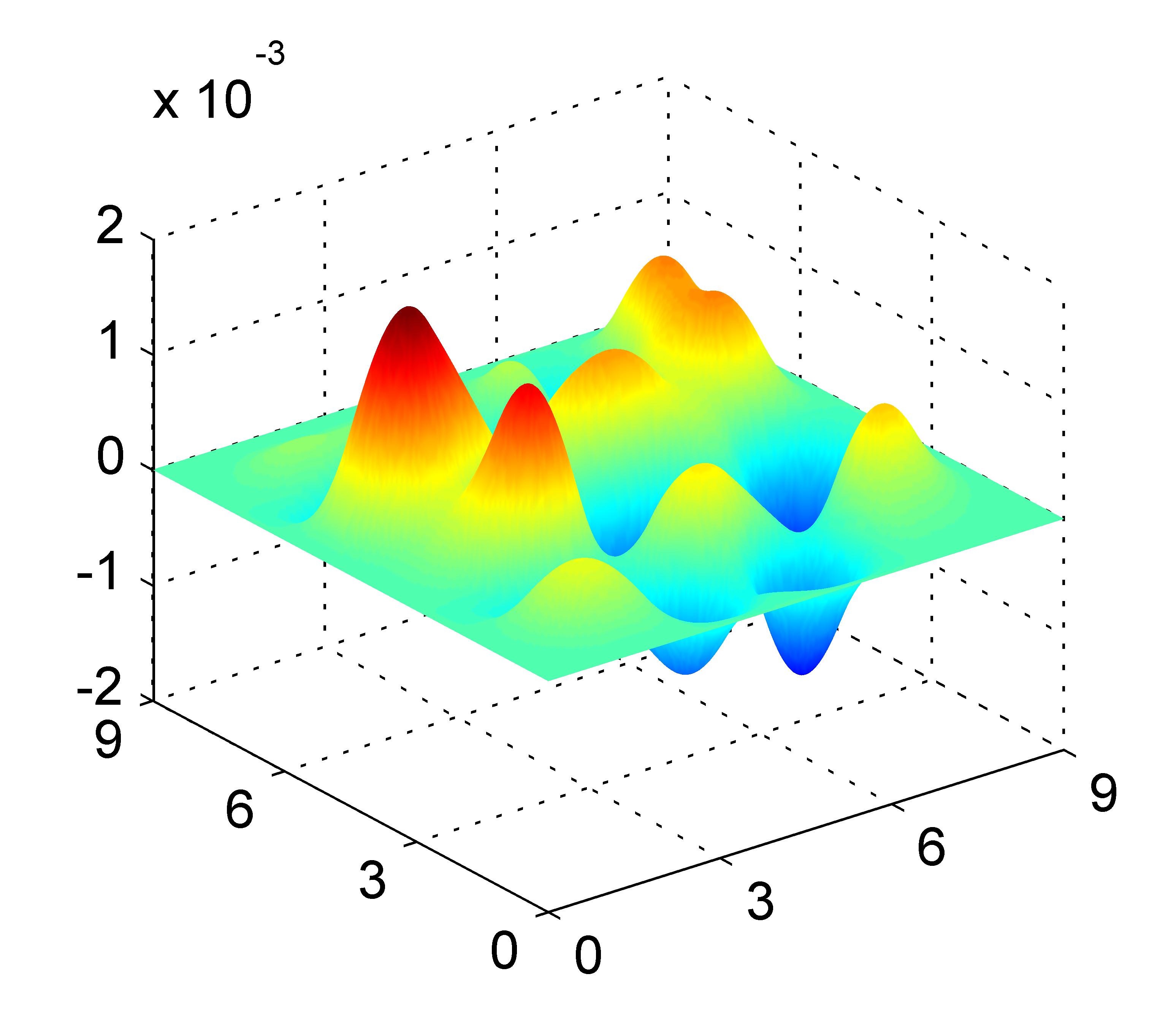}
\caption{Plotted on the left is a spline  signal in  \eqref{tensorsignal.def0} with $K_1= K_2 = 9$. In the middle and on the right are the difference between the above signal $f$ and the signal $f_{\epsilon}$ reconstructed by the MAPSET algorithm
with noise level $\epsilon = 10^{-4}$ and the sampling set $\Gamma$ being  $\Gamma_0$ and $\Gamma_1$ in
Figure \ref{sample.pic}, respectively.  
The  maximal amplitude errors  $e(\epsilon)$ in \eqref{rec.error.def}
is $0.0014$ (middle) and $0.0030$ (right),
 and  the reconstruction error $\min_{\delta \in \{-1, 1\}}\|f_{\epsilon}-\delta f\|_{\infty}$ is  $7.2567\times 10^{-4}$ (middle) and $0.0015$ (right), respectively.
}\label{signalandrec.fig}
\end{center}
\end{figure}

 The  signal $f_\epsilon$ reconstructed from the MAPSET algorithm may not provide a good approximation, up to a sign, to the original signal $f$ if the noise level $\epsilon$ is larger than a multiple of $\sqrt{M_0}/\|\Phi^{-1}\|_{2}$, cf. \eqref{stability.thm.eq1} in Theorem \ref{stability.thm}. Our numerical simulations indicate that
 the  MAPSET algorithm sometimes fails to save the phase of the original signal $f$ when $\epsilon\ge 3\times 10^{-3}$ for $\Gamma=\Gamma_0$ and  $\epsilon\ge 8\times 10^{-4}$  for $\Gamma=\Gamma_1$. 

\subsection{Nonseparable spline signals of non-tensor product type}\label{nontensor.al}
Let $M_{\Xi_Z}$ be the box spline function in  \eqref{boxspline.def} with
 $\Xi_Z=\left ( \begin{array} {cccc} 1 & 1 & 0 & 1 \\
 0 & 0 & 1 & 1\end{array}\right  )$, see \cite{deboorbook}.
 Unlike the spline function  $B_{(3,3)}$ of tensor-product type,  the shift-invariant space spanned by $M_{\Xi_Z}$ does not have the local complement property on $(0, 1)^2$, cf. Section \ref{tensor.al}. Set  $A_U:= \{(s,t): 0<s<t<1\}$ and $A_L:= \{(s, t): 0<t<s<1\}$. 
  Then the triangle regions $A_U$ and $A_L$ satisfy \eqref{necandsuf.thm.eq1}, and the  shift-invariant space spanned by $M_{\Xi_Z}$ has local complement property on $A_U$  and on $A_L$.

  For $A=A_U$ and $\phi=M_{\Xi_Z}$, the function $\Phi_{A_U}(s, t)$ in \eqref{phiA.def}  is a 5-dimensional vector-valued polynomial about $s^2, (t-s)^2, s, t-s, 1$, and
  the set $K_{A_U}$ in \eqref{KA.def} is
$
\{(0, 0), (-1, 0), (-2,0), (-1, -1), (-2, -1)\}
$. Hence
 the space spanned by the outer products of $\Phi_{A_U}(s, t)$
has dimension $13$,  and we can select a set
$\Gamma_{2, U} \subset A_U$ with cardinality $13$
   to satisfy \eqref{gamma.1}, see Figure \ref{sample_nontensor.pic}.
   Similarly, for the lower triangle region $ A_L$,  a sampling set $\Gamma_{2, L}$ with cardinality $13$ can be chosen to satisfy  \eqref{gamma.1}. 
    For our simulations, we use $$\Gamma_2=\Gamma_{2, U}\cup \Gamma_{2, L} $$
    as the sampling set contained in $A_U\cup A_L\subset (0, 1)^2$, see Figure \ref{sample_nontensor.pic}. For the above set $\Gamma_2$, the corresponding $\|\Phi^{-1}\|_2$ in  \eqref{phil2.def}  is
 $87.9420$.

 Recall $\Phi_{A_U}(s, t)$
  is a  vector-valued polynomial about  $s^2, (t-s)^2, s, t-s$ and $ 1$. Then
the matrix
  $(\Phi_{A_U}(s_i, t_i))_{1\le i\le 5}$ has full rank $5$ for almost all
  $(s_i, t_i)\in A_U, 1\le i\le 5$, and
    $(\Phi_{A_U}(s_i, t_i))_{1\le i\le 9}$ is a phase retrieval frame for almost all
  $(s_i, t_i)\in A_U, 1\le i\le 9$.
 So we can use  randomly distributed  sets
 $\Gamma_{3, U}\subset A_U$ and  $\Gamma_{3, L}\subset A_L$ with cardinality 9 that satisfy \eqref{gamma.2}, see Figure \ref{sample_nontensor.pic}.   Set $$\Gamma_3=\Gamma_{3,U}\cup \Gamma_{3, L}.$$  For the above set $\Gamma_3$, the corresponding $\|\Phi^{-1}\|_2$ in \eqref{phil2.def} is
$761.2227$.

 \begin{figure}[t] 
\begin{center}
\includegraphics[width=60mm, height=52mm]{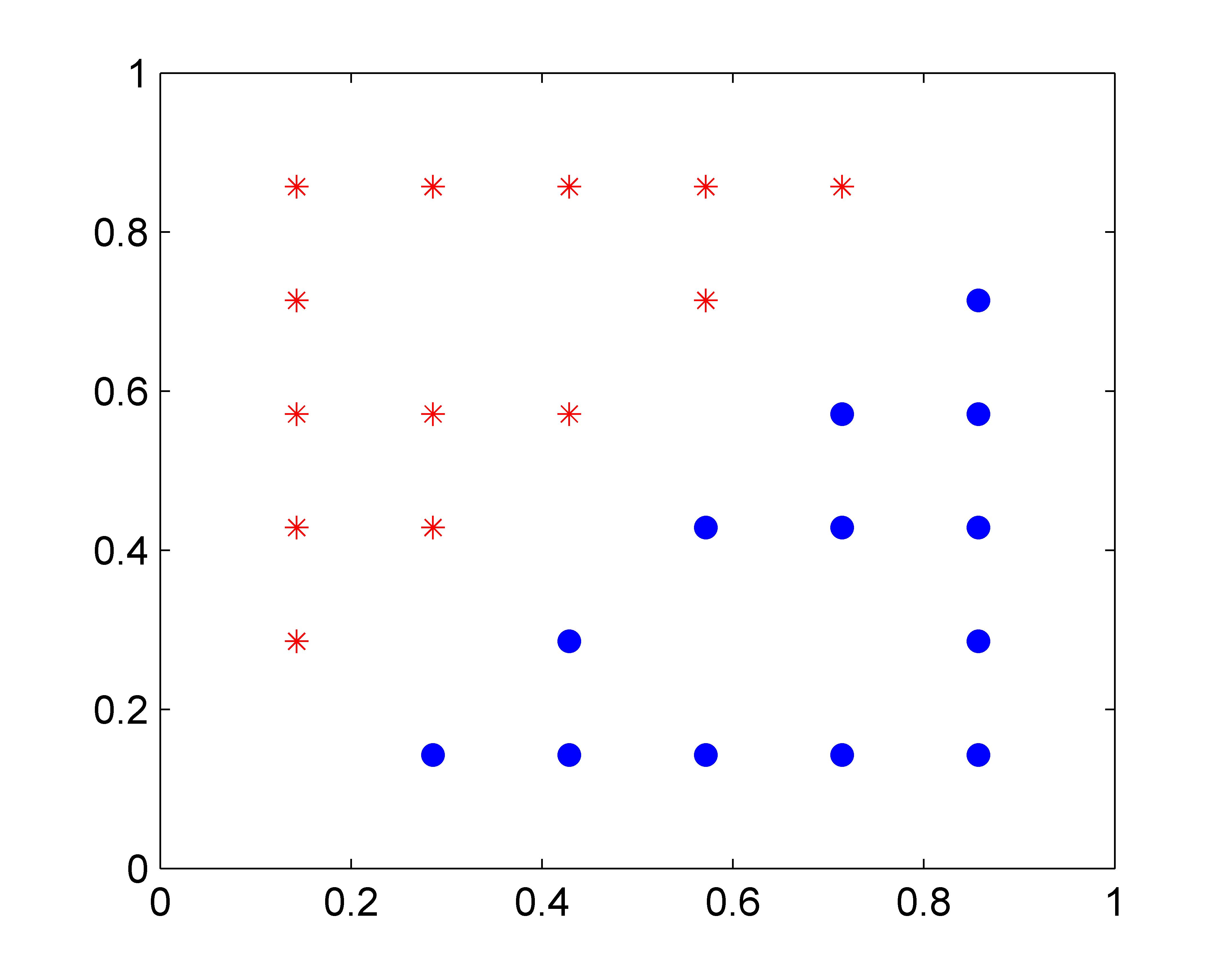}
\includegraphics[width=60mm, height=52mm]{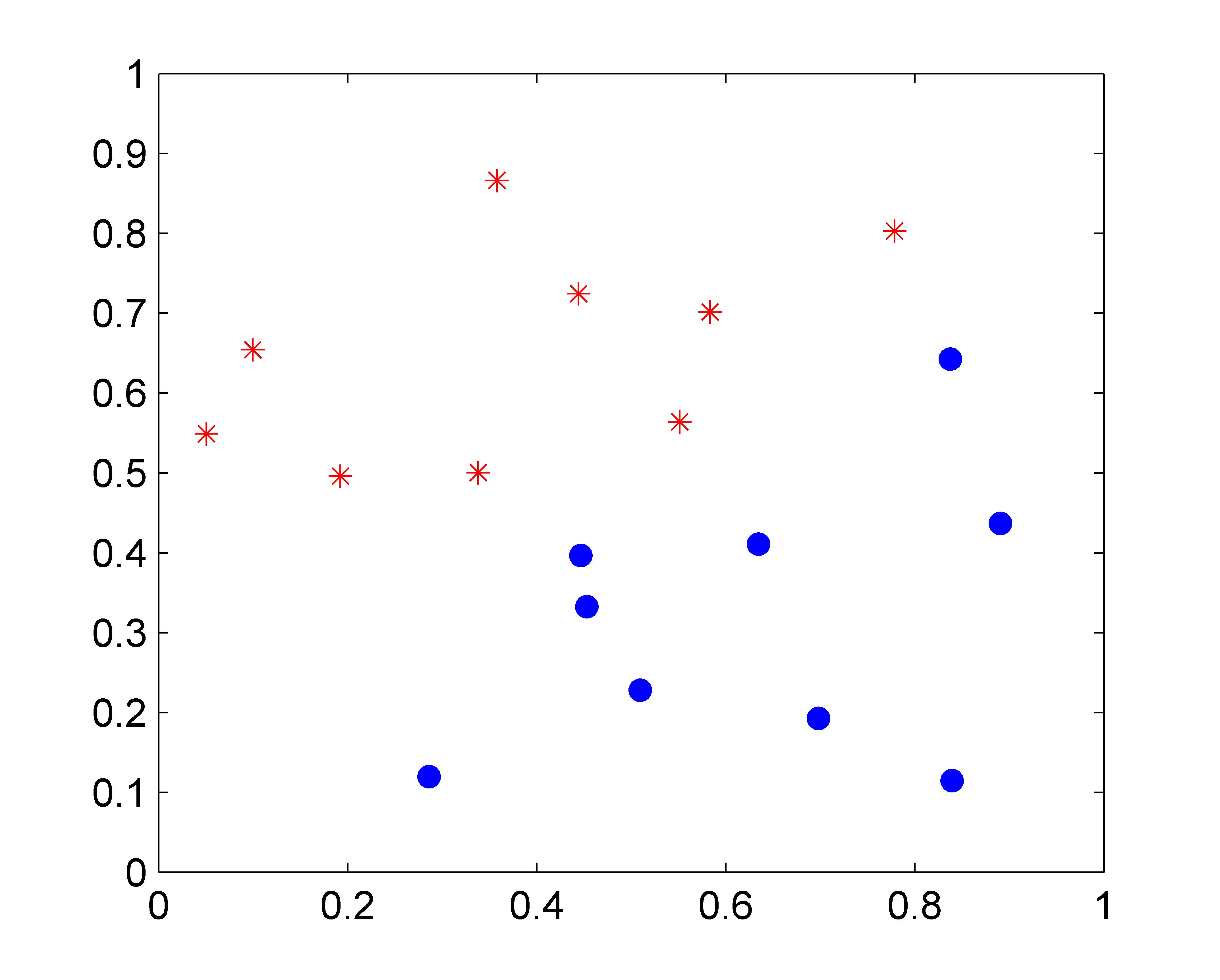}
\caption{Plotted on the left are the sampling sets $\Gamma_{2, U}\subset A_U$ (in the  red star) and   $\Gamma_{2, L}\subset A_L$
(in the blue dot).
Plotted on the right are the  random sets $\Gamma_{3,U} \subset A_U$ (in the red star) and $\Gamma_{3, L}\subset A_L$ (in the blue dot) that have cardinality $9$. The corresponding $\|\Phi^{-1}\|_2$ in \eqref{phil2.def} to the above sets is  $87.9420$ (left) and $761.2227$ (right) respectively. }\label{sample_nontensor.pic}
\end{center}
\end{figure}

In our simulations, the available  data  $z_\epsilon(y)=|f(y)|+\epsilon(y)\ge 0, y\in \Gamma+K$,  are  noisy phaseless samples of a  spline signal
  \begin{equation}\label{nontensorsignal.def0}
  f(s, t)=\sum_{0\le m\le K_1, 0\le  n\le K_2} c(m,n)M_{\Xi_Z}(s-m, t-n),  \end{equation}
taken on $\Gamma+K$, where
 $K=[0, K_1]\times [0, K_2]$ for some $1\le K_1, K_2\in \Z$,
 $\Gamma$ is either $\Gamma_2$  or $\Gamma_3$ in
Figure \ref{sample_nontensor.pic},
 amplitudes of the signal $f$ are as in \eqref{amplitude.def},
and   the additive noises $\epsilon(y)\in [-\epsilon, \epsilon], y\in \Gamma+K $,
  with noise level $\epsilon\ge 0$ are randomly selected.
Denote the signal reconstructed by the proposed MAPSET algorithm with thresholding constant $M_0=0.01$ 
 by
 \begin{equation}
 f_{ \epsilon}(s, t)=\sum_{-2\le m\le K_1, -1\le n\le K_2}c_{\epsilon}(m,n)M_{\Xi_Z}(s-m, t-n).
 \end{equation}
As in Section \ref{tensor.al}, the
reconstruction signal $f_\epsilon$ provides an approximation, up to a sign, to the original signal $f$.
 Our numerical simulation indicates that
  the MAPSET algorithm saves phases 
  in 1000 trials and the reconstruction error  $e(\epsilon)$ is about
   $O(\epsilon)$,
  provided that  $\epsilon \le 8\times 10^{-3}$ for  $\Gamma=\Gamma_2$ and $\epsilon \le 4\times 10^{-3}$  for  $\Gamma=\Gamma_3$,
  where $\sqrt{M_0}/\|\Phi^{-1}\|_{2}$ are $0.0011$ and  $1.3137 \times 10^{-4}$ respectively.  Presented in Figure \ref{signal_nontensor.fig} are
  a spline  signal  $f$ in  \eqref{nontensorsignal.def0} with $(K_1, K_2) =(9, 8)$, and the difference between the original signal $f$ and
 the reconstructed signal $f_{\epsilon}$ via the MAPSET algorithm with  noise level $\epsilon =10^{-4}$.
\begin{figure}[t] 
\begin{center}
\includegraphics[width=40mm, height=28mm]{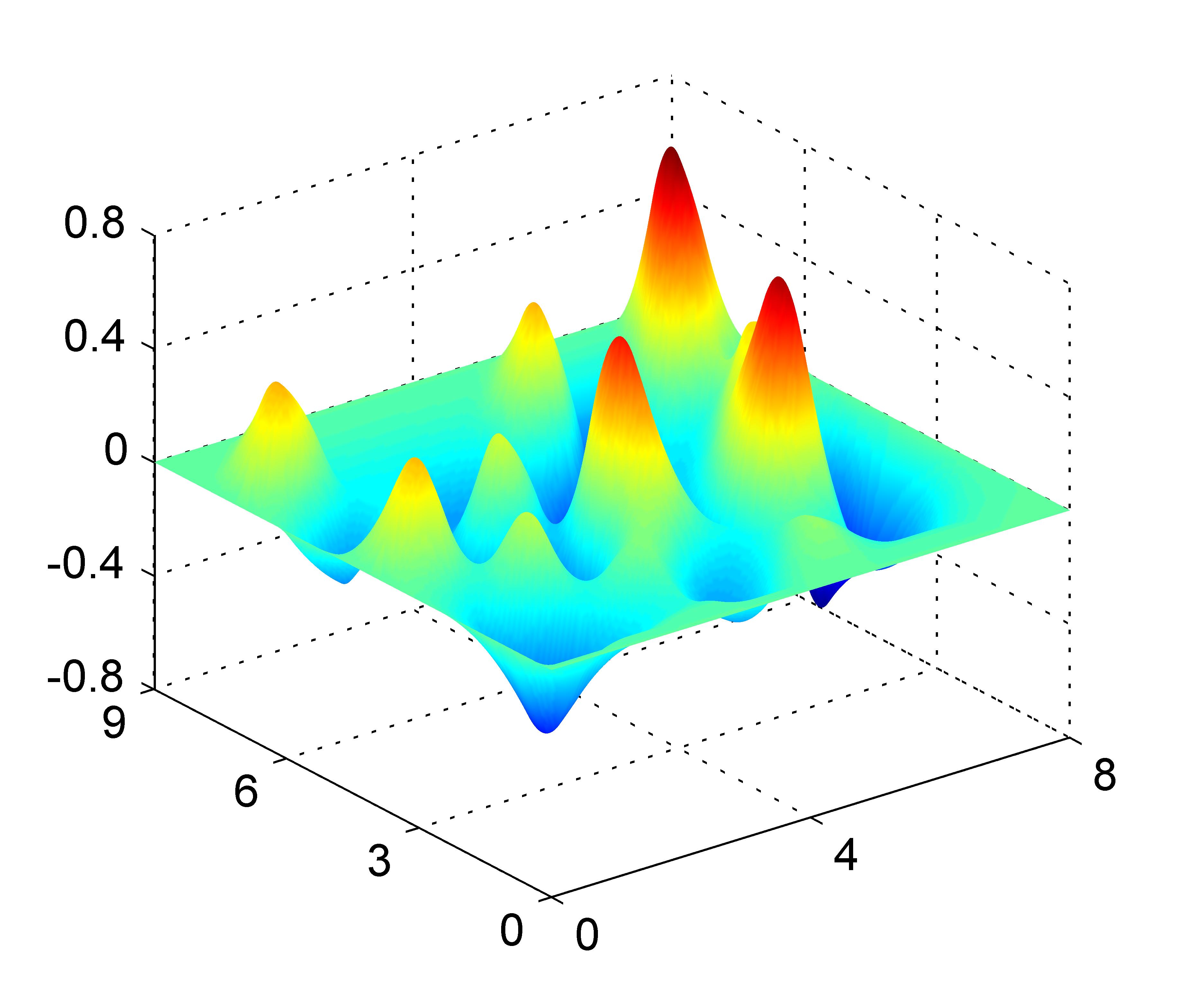}
\includegraphics[width=40mm, height=28mm]{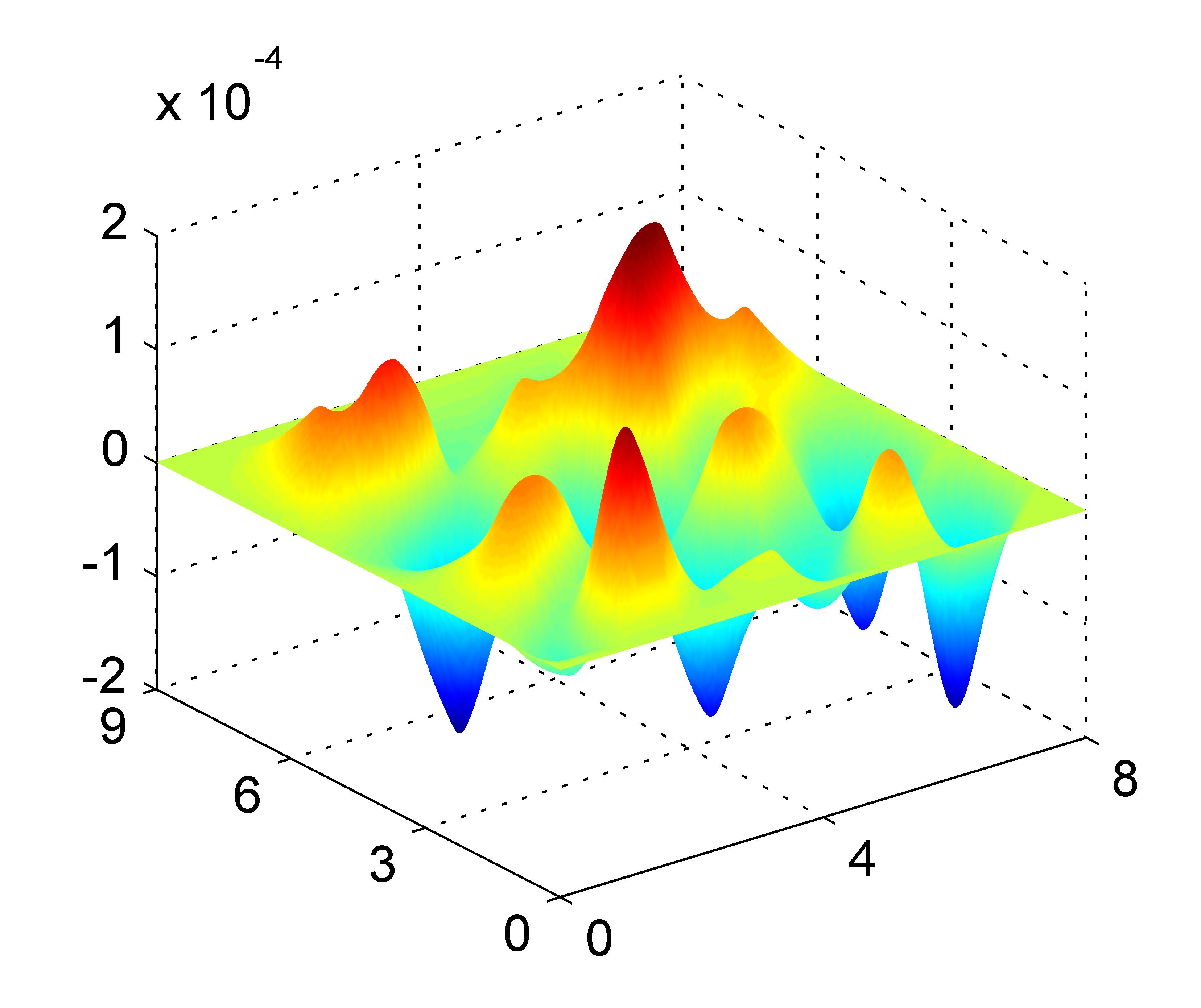}
\includegraphics[width=40mm, height=28mm]{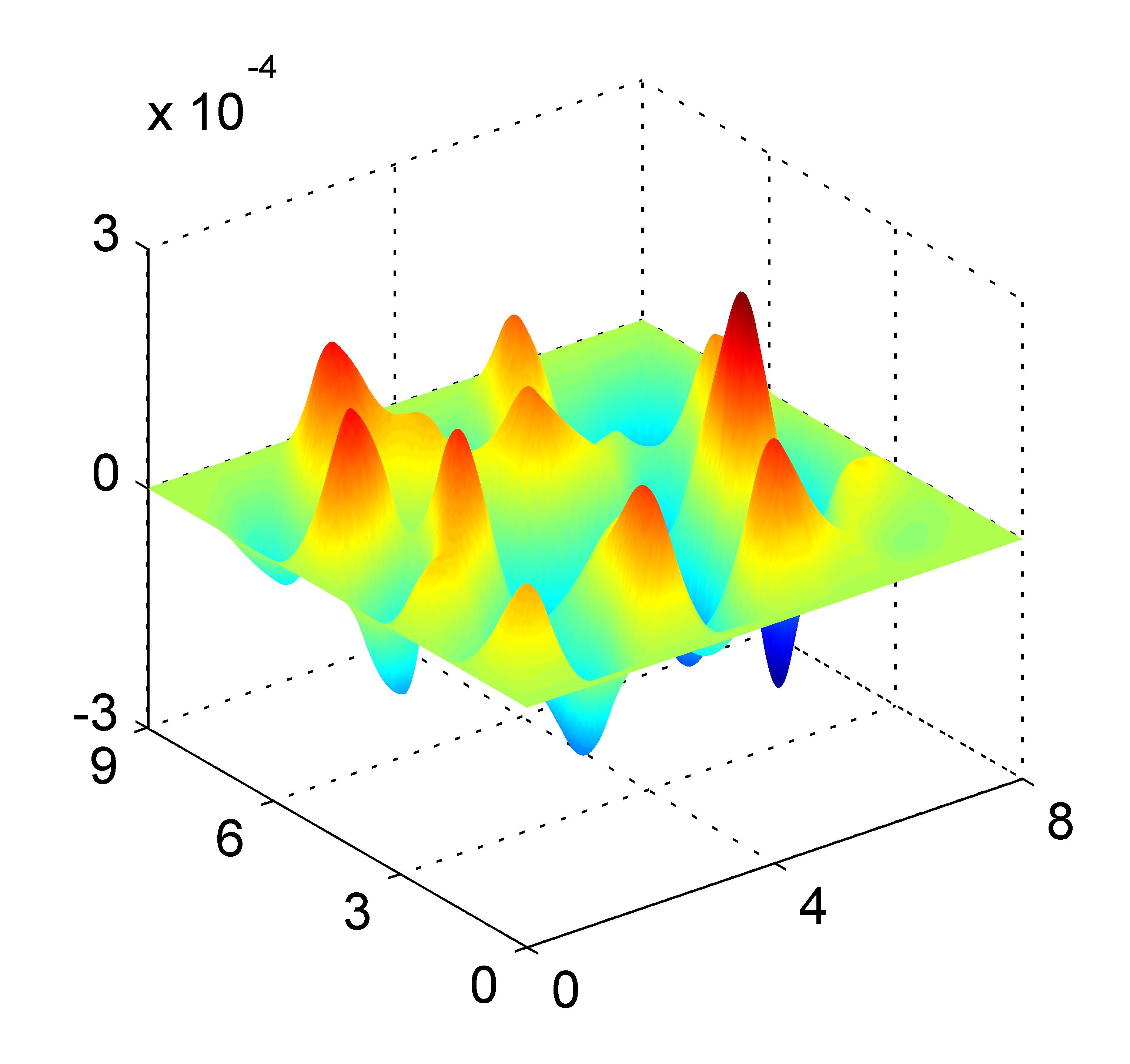}
\caption{Plotted on the left is a nonseparable signal of the form \eqref{nontensorsignal.def0},  
 where $K=[0, 9]\times [0, 8]$,
and in the middle and on the right are the difference between the above signal $f$ and the signal $f_{\epsilon}$ reconstructed by the MAPSET algorithm
with noise level $\epsilon = 10^{-4}$ and the sampling set $\Gamma$ being  $\Gamma_2$ and $\Gamma_3$ in
Figure \ref{sample_nontensor.pic}, respectively.
The  maximal amplitude error  $e(\epsilon)$ in \eqref{rec.error.def}
is $2.4922\times 10^{-4} $ (middle) and $3.8975\times 10^{-4}$ (right).  The reconstruction error $\min_{\delta \in \{-1, 1\}}\|f_{\epsilon}-\delta f\|_{\infty}$ is  $1.9660\times 10^{-4} $ (middle) 
and $ 2.9216\times 10^{-4}$ (right).  
}
\label{signal_nontensor.fig}
\end{center}
\end{figure}

As in Section \ref{tensor.al}, 
the MAPSET algorithm may not yield a good approximation to the original signal if the noisy level $\epsilon$ is not sufficient small.  Our numerical results indicate that the MAPSET algorithm sometimes fails to save the phase of the original signal $f$ when $\epsilon \ge 9 \times 10^{-3}$ for $\Gamma=\Gamma_2$ and $\epsilon \ge 5\times 10^{-3}$ for $\Gamma=\Gamma_3$.

\section{Proofs}
\label{proof.section}

In this section, we include the proofs of Theorems  \ref{necessarypr.thm},
\ref{tofinite.thm}, \ref{necandsuf.thm},  \ref{finite.cor},  \ref{sufficient.thm}, \ref{stability.thm}
and Corollary \ref{tensor.cor}.

\subsection{Proof of Theorem \ref{necessarypr.thm}}
\label{necessary.appendix}

Suppose, on the contrary, that  ${\mathcal G}_f$ is disconnected.
Let $W$ be the set of vertices in a connected component of the graph ${\mathcal G}_f$. Then
$W\ne \emptyset, V_f\backslash W\ne \emptyset$, and
there are no edges between vertices in $W$ and $V_f\backslash W$.
Write
 \begin{eqnarray}\label{necessarypr.thm.pf.eq1}
 f &\hskip-0.08in  = & \hskip-0.08in \sum_{{k}\in \Zd} c({k}) \phi(\cdot-{k})=
\sum_{k\in V_f} c({ k}) \phi(\cdot-{ k})\nonumber
 \\
 & \hskip-0.08in = & \hskip-0.08in \sum_{{ k}\in W}c({k})\phi(\cdot-{k})+\sum_{{ k}\in V_f\backslash W}c({k})\phi(\cdot-{ k}) 
 =: f_1+f_2\end{eqnarray}
where $c({k})\in \R, {k}\in \Zd$.
From the global linear independence on $\phi$ and nontriviality of the sets $W$  and $V_f\backslash W$, we obtain
\begin{equation} \label{necessarypr.thm.pf.eq2}
f_1\ne 0 \ \ {\rm  and}\  \  f_2\ne 0.\end{equation}
Combining
\eqref{necessarypr.thm.pf.eq1} and \eqref {necessarypr.thm.pf.eq2}
   with nonseparability of the signal $f$, we obtain that $f_1({x}_0)f_2({x}_0)\neq 0$ for some ${x}_0\in \Rd$.  Then by the global linear independence of   $\phi$,  there exist ${k}\in W$ and  ${k}'\in  V_f\backslash W$ such that $\phi({x}_0-{k})\neq 0$ and  $\phi({ x}_0-{k}')\neq 0$.
   Hence $({k}, {k}')$ is an edge between ${ k}\in W$ and ${ k}'\in V_f\backslash W$,  which contradicts to the construction of the set $W$. 

\subsection{Proof of Theorem \ref{tofinite.thm}}\label{tofinite.section}
A linear  space $V$ 
 on $\Rd$ is said to be {\em locally finite-dimensional} if it has finite dimensional restrictions on any bounded
open set.  The shift-invariant space
generated by a compactly supported function is locally finite-dimensional.
 The reader may refer to \cite{akram02} and references therein
on locally finite-dimensional spaces.  In this section, we will prove the following generalization of Theorem \ref{tofinite.thm}.

 \begin{theorem}\label{tofinitenew.thm}  Let
 $V$ be locally finite-dimensional shift-invariant  space of functions on $\Rd$.
Then there exists a finite set $\Gamma\subset (0, 1)^d$ such that
any nonseparable signal $f\in  V$ can be reconstructed, up to a sign,
 from its phaseless samples on $\Gamma+\Zd$.
 \end{theorem}

\begin{proof}
Let $A=(0,1)^d$ and  $V|_A$  be the space containing restrictions of all signals in $V$  on $A$.
By the shift-invariance, it suffices to find a set $\Gamma\subset A$ such that
 \begin{equation}\label{tofinitenew.pf.eq1}
|f({x})|^2=\sum_{\gamma\in \Gamma}
d_\gamma({x}) |f({\gamma})|^2, \ x\in A
\end{equation}
 hold for all  $f\in V$.
  By the assumption on $V$,   $V|_A$ is a finite-dimensional space. Let $ g_n\in V, 1\le n\le N$, be a basis for $V|_A$, and  $W$ be the linear space generated by symmetric matrices
$$G({x}):= \big(g_n({x}) g_{n'}({x})\big)_{1\le n, n'\in N}, \ {x}\in A.$$
Then there exists a finite set $\Gamma\subset A$ with cardinality at most $N(N+1)/2$ such that
$G({\gamma}), \gamma\in \Gamma$, is a basis  for the space $W$.
This implies that  for any ${x}\in A$ there exist $d_\gamma({x}), \gamma\in \Gamma$, such that
$$G(x)=\sum_{\gamma\in \Gamma} d_\gamma({x}) G({\gamma}), \ {x}\in A.$$
For any $f\in V$, we write $f=\sum_{n=1}^N c_n g_n$ on $A$.
Then
\begin{eqnarray*}
|f({x})|^2 & = &  \Big|\sum_{n=1}^N c_n g_n({ x})\Big|^2=
\sum_{n,n'=1}^N c_n c_{n'} g_n({x}) g_{n'}({x})\\
& = &
\sum_{n,n'=1}^N c_n c_{n'} \Big(\sum_{\gamma\in \Gamma}d_\gamma(x)  g_n(\gamma) g_{n'}(\gamma)\Big) =  \sum_{\gamma\in \Gamma}
d_\gamma(x) |f({\gamma})|^2,\ x\in A.
\end{eqnarray*}
This proves \eqref {tofinitenew.pf.eq1} and hence completes the proof.
\end{proof}

\subsection{Proof of Theorem \ref{necandsuf.thm}}
\label{necandsuf.appendix}

The implication (iii)$\Longrightarrow$(i) is obvious, while the  implication (i)$\Longrightarrow$(ii) follows from Theorem \ref{necessarypr.thm} with its proof given in Section  \ref{necessary.appendix}.
 Now it remains to prove  
  (ii)$\Longrightarrow$(iii).
Let $\Gamma_m, 1\le m\le M$, be finite sets constructed in Proposition \ref{finiteset.pr} 
 with the set $A$ and the space $V$ replaced by $A_m$ and $V(\phi)$ respectively, and set $\Gamma=\cup_{m=1}^M \Gamma_m$.

 Let
$f,  g\in V(\phi)$  satisfy
\begin{equation}\label{necandsuf.pf.eq-1}
|g(y)|= |f(y)|\ {\rm  for\ all} \ y\in  \Gamma+\Zd.\end{equation}
Then it suffices to prove that
\begin{equation}\label{necandsuf.pf.eq0} g=\epsilon f
\end{equation}
for some $\epsilon \in \{-1, 1\}$.
By Proposition \ref{finiteset.pr}  and the shift-invariance of the linear space $V(\phi)$,
\begin{equation*}
|g(x+l)|= |f(x+l)|,\   x\in A_m\end{equation*}
where ${l}\in \Zd$ and $1\le m\le M$.  This, together  shift-invariance of the linear space $V(\phi)$
and local complement property on $A_m, 1\le m\le M$,
implies the existence of $\epsilon_{ l, m}\in \{-1, 1\}$ 
such that
\begin{equation}\label{necandsuf.pf.eq1}
g( x)=\epsilon_{ l, m} f(x), \  x\in A_m+l.
\end{equation}
 Write
$f=\sum_{{ k}\in \Zd} c({k}) \phi(\cdot-k)$ and
 $g=\sum_{{k}\in \Zd} d({ k}) \phi(\cdot-{k})\in V(\phi)$.
 Then it follows from  \eqref{necandsuf.pf.eq1} and local linear independence on $A_m$ for the generator $\phi$
 that
 \begin{equation}\label{necandsuf.pf.eq2}
 d({k}'+{ l})= \epsilon_{{ l}, m} c({ k}'+{ l}), \ {k}'\in K_{A_m},
 \end{equation}
 where $K_{A_m}$ is given in \eqref{KA.def}. Hence the proof of \eqref{necandsuf.pf.eq0} reduces to showing
 \begin{equation}\label{necandsuf.pf.eq3}
 \epsilon_{{ l}, m}=\epsilon
 \end{equation}
 for all ${l}\in \Zd$ and $1\le m\le M$ so that ${ k}'+{l}\in V_f$
 for some $k' \in K_{A_m}$.

 Recall that $c({ k})\ne 0$ for all ${k}\in V_f$.
Then by \eqref{necandsuf.pf.eq2}  there exist $\delta_{ k}\in \{-1, 1\}$ for all ${ k}\in V_f$
 such that
\begin{equation} \label{necandsuf.pf.eq4}
 \epsilon_{{ l}, m}=\delta_{k}
\end{equation}
  for all ${ l}\in \Zd$ and $1\le m\le M$ so that ${ k}={ k}'+{l}\in V_f$
 for some $k' \in K_{A_m}$. By \eqref{necandsuf.pf.eq4} and the connectivity of the graph ${\mathcal G}_f$,
  the proof of \eqref{necandsuf.pf.eq3}
reduces further to proving
 \begin{equation}\label{necandsuf.pf.eq5} \delta_{ k}=\delta_{\tilde{ k}} \end{equation}
 for all edges $({ k}, \tilde{ k})$ of the graph ${\mathcal G}_f$.

For an edge $({ k}, \tilde{k})$ of the graph ${\mathcal G}_f$,
we have that
$$S:=\{{x}: \  \phi({x}-{k})\phi({ x}-\tilde{k})\ne 0\}\ne \emptyset.$$
Then
there exist $1\le m\le M$  by \eqref{necandsuf.thm.eq2} and \eqref{necandsuf.thm.eq1} such that
$S\cap (A_m+{k})\ne \emptyset$. Thus
  ${ k}, \tilde{k}\in K_{A_m}+{k}$,  which together with
  \eqref{necandsuf.pf.eq2} and \eqref{necandsuf.pf.eq4} implies that
$\delta_{k}=\epsilon_{ k, m} =\delta_{\tilde{k}}$.
Hence \eqref{necandsuf.pf.eq5} holds.  This completes the proof.

\subsection{Proof of Corollary \ref{tensor.cor}} \label{tensor.proof}
As  restriction of  signals in $V(B_{\bf N})$ on $(0, 1)^d$
 are polynomials of finite degree, the space $V(B_{\bf N})$ has the complement property on $(0, 1)^d$.
Set ${\bf n}=(n, \ldots, n)$ for $n\in \Z$.
 It is observed that the function
$\Phi_{(0,1)^d}$ in \eqref{phiA.def}
is  a vector-valued polynomial of degree ${\bf N}-{\bf 1}$,
and  its outer product
$\Phi_{(0,1)}(x) \Phi_{(0,1)}(x)^T$ is
 a matrix-valued polynomial of degree $2{\bf N}-{\bf 2}$.
Therefore
 $\Phi_{(0,1)}(y) \Phi_{(0,1)}(y)^T, y\in X_1\times \cdots\times X_d$,
 is a spanning set of the space spanned by $\Phi_{(0,1)}(x) \Phi_{(0,1)}(x)^T, x\in (0,1)^d$.
 This together with Theorem \ref{necandsuf.thm} completes the proof.

\subsection{Proof of Theorem \ref{finite.cor}}
Let $f=\sum_{k\in \Zd}c(k)\phi(\cdot-k)$ and $g=\sum_{k\in \Zd}d(k)\phi(\cdot-k)$
satisfy
$$|g(y)|=|f(y)| \ \ {\rm for \ all}\ \  y\in \Gamma'+\Zd,$$
where $\Gamma'=\cup_{m=1}^M \Gamma_m^\prime$ is given in \eqref{gamma.2}.
 Take  $l\in \Zd$ and  $ 1\le m\le M$. Then
      $$\Big|\sum_{k\in K_{A_m}+l}d(k)\phi(\gamma_m+l-k)\Big|=\Big|\sum_{k\in K_{A_m}+l}c(k)\phi(\gamma_m+l-k)\Big|\quad {\rm for \ all} \ \gamma\in \Gamma_m^\prime.$$
By the assumption on $\Phi_{A_M}(\gamma'), \gamma\in \Gamma_m', 1\le m\le M$, and the shift-invariance of the linear space $V(\phi)$,  there exists $\epsilon_{m,l}\in \{1,-1\}$ such that
 \begin{equation*}
      d(k)=\epsilon_{m,l}c(k), \ k\in K_{A_m}+l.
      \end{equation*}
 Following the same argument as the one used in  the implication i)$\Longrightarrow$iii) in Theorem \ref{necandsuf.thm},  we can find $\epsilon\in \{-1, 1\}$ such that
    $\epsilon_{l,m}=\epsilon $
      for all $l\in \Zd$ and $1\le m\le M$. This completes the proof.

\subsection{Proof of Theorem \ref{sufficient.thm}}
By Proposition \ref{llilocalcomplement.pr},
 there are open sets $A_1, \ldots, A_M$ satisfy the requirements
in Theorem \ref{necandsuf.thm}. Then the conclusion in  Theorem \ref{sufficient.thm}
follows from Theorem \ref{necandsuf.thm}.

\subsection{Proof of Theorem \ref{stability.thm}}
Given $\Gamma\subset \Rd$ and $f=\sum_{k\in \Zd} c(k) \phi(\cdot-k)$, we define
\begin{equation}\label{graphG2.def}
\widetilde {\mathcal G}_{f, \Gamma}=(V_f, E_{f, \Gamma}),\end{equation}
where
$(k, k')\in E_{f, \Gamma}$ only if
$\phi({ y}-k) \phi({ y}-k')\ne 0$
for some ${y}\in \Gamma +\Zd$.  
To prove Theorem \ref{stability.thm}, we need a lemma about  the graph ${\mathcal G}_f$.

\begin{lemma}\label{stability.lem}  Let  $\phi$,  $A_m$ and $\Gamma_m, 1\le m\le M$, be as in Theorem \ref{stability.thm}. Set $\Gamma=\cup_{m=1}^M \Gamma_m$.
 Then for any  $f\in V(\phi)$, the graph ${\mathcal G}_f$ in \eqref{graphG.def} and
 $\widetilde {\mathcal G}_{f, \Gamma}$ in \eqref{graphG2.def} are the same,
\begin{equation} \label{stability.lem.eq1}
{\mathcal G}_f=\widetilde {\mathcal G}_{f, \Gamma}.
\end{equation}
\end{lemma}

\begin{proof} Clearly it suffices to prove that an edge 
in ${\mathcal G}_f$ is also an edge in $\widetilde {\mathcal G}_{f, \Gamma}$.
Suppose, on the contrary, that there exists an edge $(k, k')\in E_f$
such that
\begin{equation}\label{stability.lem.pf.eq1}
\phi({y}-k) \phi({ y}-k')= 0 \ {\rm for \ all}\ { y}\in \cup_{m=1}^M \Gamma_m+\Zd.
\end{equation}
Define
\begin{equation}
\label{stability.lem.pf.eq2}
A=\{{ x}\in \Rd: \ \phi({ x}-k)\phi({ x}-k')\ne 0\} \ne \emptyset.
\end{equation}
   By \eqref{necandsuf.thm.eq1},  there exist $l_0\in \Zd$ and $1\le m_0\le M$ such that
\begin{equation}\label{stability.lem.pf.eq3}
A\cap (A_{m_0}+l_0)\ne \emptyset. \end{equation}
Set $g_{\pm}( x) = \phi({ x}+l_0-k)\pm  \phi({x}+l_0-k')$.
Then it follows from \eqref{stability.lem.pf.eq1} that
$$|g_{\pm}(\gamma)| 
=|\phi(\gamma+l_0-k)|+|\phi( \gamma+l_0-k')|, \ \gamma\in \Gamma_{m_0}.
$$
 By
 the construction of the set $\Gamma_{m_0}$, 
  we get either
$g_{+} =g_{-}$ or $g_{+}= -g_{-}$ on $A_{m_0}$. Therefore
either $\phi({x}+l_0-k)\equiv 0$ on $A_{m_0}$ or $\phi({ x}+l_0-k')\equiv 0$ on $A_{m_0}$. This contradicts to the construction of set $A$ in \eqref{stability.lem.pf.eq2} and \eqref{stability.lem.pf.eq3}.
\end{proof}

Now, we continue the proof of Theorem \ref{stability.thm}.

\begin{proof} [Proof of Theorem \ref{stability.thm}]
 Take $l \in \Zd$ and $1\le m\le M$.
For $\gamma\in \Gamma_m$ there exists $\tilde \delta_{\gamma, l, m}\in \{-1, 1\}$ such that
\begin{eqnarray}\label{stability.thm.pf.eq1}
 & &\Bigg(\sum_{\gamma\in \Gamma_m} \Bigg|\sum_{k\in l+\Omega_{m}} (c_{ \epsilon, l; m}(k)-\tilde \delta_{ \gamma, l; m} c(k)) \phi(\gamma+l-k)\Bigg|^2\Bigg)^{1/2}
 \nonumber\\
 & = & \Bigg(\sum_{\gamma\in \Gamma_m}\Bigg| \Big|\sum_{k\in l+\Omega_{m}} c_{ \epsilon, l; m}(k)\phi(\gamma+l-k)\Big|\nonumber\\
 & & \qquad \qquad\quad -
\Big|\sum_{k\in l+\Omega_{m}} c(k) \phi(\gamma+l-k)\Big|\Bigg|^2\Bigg)^{1/2}\nonumber\\
 & \le &
 \Bigg(\sum_{\gamma\in \Gamma_m}\Bigg| \Big|\sum_{k\in l+\Omega_{m}} c_{\epsilon, l; m}(k)\phi(\gamma+l-k)\Big|-
 z_{\epsilon} (\gamma+l)\Bigg|^2\Bigg)^{1/2}\nonumber\\
 &  & +  \Bigg(\sum_{\gamma\in \Gamma_m}\Bigg|
\Big|\sum_{k\in l+\Omega_{m}} c(k) \phi(\gamma+l-k)\Big|-z_{\epsilon} (\gamma+l)\Bigg|^2\Bigg)^{1/2}\nonumber\\
& \le & 2 \Bigg(\sum_{\gamma\in \Gamma_m}\Bigg|
\Big|\sum_{k\in l+\Omega_{m}} c(k) \phi(\gamma+l-k)\Big|-z_{\epsilon} (\gamma+l)\Bigg|^2\Bigg)^{1/2}\nonumber\\
&\le &  2 \sqrt{\# \Gamma_m}  \|\epsilon\|_{\infty}\le  2  \sqrt{\# \Gamma} \|\epsilon\|_{\infty} ,
\end{eqnarray}
where the second inequality holds by \eqref{meps.def2}
and the last inequality follows from
$$z_{ \epsilon} (\gamma+l)=\Big|\sum_{k\in l+\Omega_{m}} c(k) \phi(\gamma+l-k)\Big|+\epsilon({\gamma}+l), \ {\gamma}\in \Gamma_m.$$

 From the phase retrievable frame property for $\big(\phi(\gamma-k)\big)_{k\in K_{A_m}}, \gamma\in \Gamma_m$, we obtain that
 \begin{equation}\label{omegam.equiv}
 \Omega_m=K_{A_m},\  1\le m\le M.
 \end{equation}
Let $A_{l, m}=\{\gamma\in \Gamma_m: \  \delta_{\gamma, l; m}=1\}$.
 This together with \eqref{omegam.equiv} and the phase retrievable frame assumption that either $\big(\phi(\gamma-k)\big)_{k\in \Omega_{m}}, \gamma\in A_{l, m}$ or
$\big(\phi(\gamma-k)\big)_{k\in \Omega_{m}},   \gamma\in \Gamma_m\backslash A_{l, m}$ is a spanning set for $\R^{\# \Omega_{m}}$. 
This together with
\eqref{stability.thm.pf.eq1} implies that
\begin{equation} \label{stability.thm.pf.eq2}
 \Big(\sum_{k\in l+\Omega_{m} } \big|c_{\epsilon, l; m}(k)- \tilde\delta_{l,m} c(k)\big|^2\Big)^{1/2} 
  \le   2 \|\Phi^{-1}\|_2 \sqrt{\# \Gamma} \|\epsilon\|_\infty
\end{equation}
for some sign $\tilde\delta_{l, m}\in \{-1, 1\}$.

Now we show that phases of  ${ c}_{\epsilon, l; m}, l\in \Zd, 1\le m\le M$,
can be adjusted so that \eqref{meps.def13} holds.
Let $\tilde\delta_{l, m}, l\in \Zd, 1\le m\le M$, be as in \eqref{stability.thm.pf.eq2}.
Then for any $l, l'\in \Zd$ and  $1\le m, m'\le M$,  set
$\Omega_{l, m; l',  m'}=(\Omega_m+l)\cap (\Omega_{l'}+m')$. Then
\begin{eqnarray}\label{stability.thm.pf.eq3}
 \quad  & & \langle \tilde\delta_{l, m}{ c}_{\epsilon, l; m},
\tilde\delta_{l', m'}{c}_{\epsilon, l'; m'}\rangle
  =  \sum_{k\in \Omega_{l, m; l',  m'}}
 \tilde\delta_{l, m} \tilde\delta_{l', m'} c_{ \epsilon, l; m}(k)
 c_{ \epsilon, l'; m'}(k)
 \nonumber\\
  \qquad  & \ge  &     \sum_{k\in \Omega_{l, m; l', m'}}|c(k)|^2
-
  \sum_{k\in \Omega_{l, m; l',  m'}}
|c(k)|  | \tilde \delta_{l', m'}c_{ \epsilon, l', m'}(k)-c(k)| \nonumber\\
   & & - \sum_{k\in \Omega_{l, m;  l', m'}}
 |\tilde\delta_{l, m} c_{\epsilon,  l; m}(k)- c(k)|
|c(k))|\nonumber\\
& & - \sum_{k\in \Omega_{l, m; l', m'}}
 |\tilde\delta_{l, m} c_{ \epsilon, l; m}(k)- c(k)|
  |\tilde \delta_{l', m'}c_{\epsilon, l'; m'}(k)-c(k)|\nonumber\\
  &\ge  &  -\frac{1}{2}
  \sum_{k\in \Omega_{l, m; l',  m'}}
  \Big(| \tilde\delta_{l', m'}c_{\epsilon, l'; m'}(k)-c(k)|^2 +
 |\tilde\delta_{l, m} c_{\epsilon,  l; m}(k)- c(k)|^2\Big)
\nonumber\\
& & - \sum_{k\in \Omega_{l, m; l',  m'}}
 |\tilde\delta_{l, m} c_{\epsilon, l; m}(k)- c(k)|
  |\tilde \delta_{l', m'}c_{\epsilon, l'; m'}(k)-c(k)|\nonumber\\
 \qquad  & \ge &  -8 \|\Phi^{-1}\|_2^2  \#\Gamma  \|\epsilon\|_\infty^2\ge -M_0,
\end{eqnarray}
where  the third inequality follows from \eqref{stability.thm.pf.eq2}
and the last inequality holds by the assumption \eqref{stability.thm.eq1}
on the noise level $\|\epsilon\|_\infty$ and the threshold constant $M_0$.

 The phase adjustments in \eqref{meps.def13}  for ${c}_{\epsilon, l, m}, l\in \Zd, 1\le m\le M$,
are non-unique.  Next we show  that they are
essentially the phase adjustments in \eqref{stability.thm.pf.eq3}, i.e.,
 for any  phase adjustments  $\delta_{l, m} \in \{-1, 1\}$  in \eqref{meps.def13} 
 there exists $\delta\in \{-1, 1\}$ such that
\begin{equation}\label{stability.thm.pf.eq4}
 \delta_{l, m} c(k)= \delta \tilde\delta_{l, m} c(k)\quad  {\rm for \ all} \ k\in l+\Omega_{m}.
\end{equation}

To prove \eqref{stability.thm.pf.eq4}, we claim that
\begin{equation}\label{stability.thm.pf.eq5}
\tilde \delta_{l, m}/ \delta_{l, m} =  \delta_{l', m'}/ \tilde \delta_{l', m'}
\end{equation}
for all  $(l, m)$ and $(l', m')$ with $\Omega_{l, m; l', m'}\cap V_f\ne\emptyset$.
Suppose on the contrary that \eqref{stability.thm.pf.eq5} does not hold.
Then
$$\langle  \delta_{l, m}{ c}_{\epsilon, l; m},
\delta_{l', m'}{c}_{\epsilon, l'; m}\rangle=-\langle \tilde \delta_{l, m}{ c}_{ \epsilon, l; m},
\tilde \delta_{l', m'}{ c}_{ \epsilon, l'; m}\rangle.
$$
Therefore
\begin{eqnarray*}\label{globalphase.eq}
  & & \langle  \delta_{l, m}{c}_{\epsilon, l; m},
 \delta_{l', m'}{c}_{\epsilon, l'; m}\rangle
 \nonumber\\
 & \le  &    - \sum_{k\in \Omega_{l, m;  l',  m'}}|c(k)|^2
+
  \sum_{k\in \Omega_{l, m;  l',  m'}}
|c(k)|  | \tilde\delta_{l', m'}c_{ \epsilon, l'; m}(k)-c(k)| \nonumber\\
   & & + \sum_{k\in \Omega_{l, m;  l',  m'}}
 |\tilde\delta_{l, m} c_{\epsilon,  l, m}(k)- c(k)|
|c(k))|\nonumber\\
& & + \sum_{k\in \Omega_{l, m;  l',  m'}}
 |\tilde\delta_{l, m} c_{\epsilon, l; m}(k)- c(k)|
  | \tilde\delta_{l', m'}c_{\epsilon, l', m}(k)-c(k)|\nonumber\\
  &\le  &  - \sum_{k\in \Omega_{l, m;  l', m'}}|c(k)|^2
 + 4 \sqrt{\# \Gamma} \|\Phi^{-1}\|_2
       \Big(\sum_{k\in \Omega_{l, m;  l',  m'}}|c(k)|^2\Big)^{1/2} \|\epsilon\|_\infty
       \nonumber\\
       & &
   +   4  \# \Gamma \|\Phi^{-1}\|_2^2
  \|\epsilon\|_\infty^2
\nonumber\\
  &\le  &  - \sum_{k\in \Omega_{l, m;  l', m'}}|c(k)|^2
 + \Big(2M_0\sum_{k\in \Omega_{l, m;  l', m'}}|c(k)|^2\Big)^{1/2}
  +\frac{M_0}{2} < -M_0,
\end{eqnarray*}
where the second inequality follows from \eqref{stability.thm.pf.eq2}, and the third and fourth inequalities hold by \eqref{stability.thm.eq0} and \eqref{stability.thm.eq1}.  This contradicts to the requirement \eqref{meps.def13} for the phase adjustment
and hence completes the proof of the Claim \eqref{stability.thm.pf.eq5}.

By \eqref{stability.thm.pf.eq5}, for any $k\in V_f$ there exist $\delta_k\in \{-1, 1\}$ such that
\begin{equation}\label{stability.thm.pf.eq6}
 \delta_{l, m} c(k)= \delta_{k}\tilde \delta_{l, m} c(k)\quad  {\rm for \ all} \ k\in l+\Omega_{m}.
\end{equation}
Let $(k_1, k_2)$ be an edge in ${\mathcal G}_f$. By Lemma \ref{stability.lem}
there exist $l\in \Zd$ and $1\le m\le M$ such that
$k_1, k_2\in \Omega_m+l$. Therefore
\begin{equation*}
 \delta_{l, m} c(k_1)= \delta_{k_1} \tilde\delta_{l, m} c(k_1)
\ {\rm and}\
 \delta_{l, m} c(k_2)= \delta_{k_2} \tilde\delta_{l, m} c(k_2)
\end{equation*}
by \eqref{stability.thm.pf.eq6}.
This implies that $\delta_{k_1}= \delta_{k_2}$ for any edge $(k_1, k_2)$ in ${\mathcal G}_f$. Combining it with the connected of the graph ${\mathcal G}_f$, we can find
  $\delta\in \{-1, 1\}$ such that
\begin{equation} \label{stability.thm.pf.eq7}
\delta_k=\delta\ \ {\rm for\ all} \ k\in V_f.
\end{equation}
Combining \eqref{stability.thm.pf.eq6} and \eqref{stability.thm.pf.eq7} proves
\eqref{stability.thm.pf.eq4}.

By  \eqref{stability.thm.pf.eq2}
and  \eqref{stability.thm.pf.eq4}, we obtain
\begin{eqnarray}  \label{stability.thm.pf.eq8}
|d_{\epsilon}(k)-\delta {c}(k)|
 & \le &
\frac{\sum_{m=1}^M\sum_{l\in \Zd} |  \delta_{l, m} { c}_{\epsilon, l; m}(k)-\delta { c}(k)| }{\sum_{m=1}^M\sum_{ l\in \Zd} \chi_{l+\Omega_{m}(k)}}\nonumber\\
& = &
\frac{\sum_{m=1}^M\sum_{l\in \Zd} |  { c}_{\epsilon, l; m}(k)- \tilde \delta_{l, m} { c}(k)| }{\sum_{m=1}^M\sum_{l\in \Zd} \chi_{l+\Omega_{ m}(k)}}\nonumber\\
& \le & 2 \sqrt{\# \Gamma} \|\Phi^{-1}\|_2 \|\epsilon\|_\infty, \ k\in \Zd.
\end{eqnarray}
This together with
\eqref{stability.thm.eq1-} and \eqref{stability.thm.eq1}
implies that
\begin{equation} \label{stability.thm.pf.eq9}
|d_{\epsilon}(k)|\ge  \frac{3}{2} \sqrt{M_0} \ \ {\rm for \ all} \ k\in V_f,
\end{equation}
and
\begin{equation}  \label{stability.thm.pf.eq10}
|d_{\epsilon}(k)|\le \frac{1}{2} \sqrt{M_0}  \ \ {\rm for \ all} \ k\not\in V_f.
\end{equation}
Combining \eqref{thresholding}, \eqref{stability.thm.pf.eq8}, \eqref
{stability.thm.pf.eq9} and \eqref{stability.thm.pf.eq10}
 completes the proof of the desired error estimate \eqref{stability.thm.eq3}
 and \eqref{stability.thm.eq3+}.
\end{proof}

\begin{appendix}

\section{Local complement property}
\label{localcomplementproperty.appendix}

A linear  space $V$ 
 on $\Rd$ is said to be {locally finite-dimensional} if it has finite-dimensional restrictions on any bounded
open set. Examples of locally finite-dimensional spaces include the space of polynomials of finite degrees, the shift-invariant space 
generated by  finitely many compactly supported functions, and their linear subspaces. The reader may refer \cite{akram02} and references therein
on locally finite-dimensional spaces.  In this section, we consider the local complement property for a
locally finite-dimensional space, cf. Definition \ref{lcp.def0}.

\begin{definition} {\rm
Let $V$ be a linear space of real-valued continuous signals on $\Rd$, and $A\subset \Rd$. We say that $V$ has {\em local complement property on $A$} if for any $A'\subset A$
there does not exist $f, g\in V$ such that $f, g\not\equiv 0$ on $A$, $f\equiv 0$ on $A'$ and $g\equiv 0$ on $A\backslash A'$.
}
\end{definition}

In the following theorem, we establish the equivalence between the local complement property on a bounded open set and  complement property for ideal sampling functionals on a finite subset.

  \begin{theorem} \label{AtoGamma.thm}
 Let $A$ be a bounded open set
 and $V$ be a locally finite-dimensional space of real-valued continuous signals on $\Rd$.
  Then $V$ has the local complement property on $A$ if and only if there exists a finite set $\Gamma\subset  A$ such that for any $\Gamma'\subset \Gamma$
  either
  there does not exist
 $f\in V$ satisfying
\begin{equation}f \not\equiv 0 \ \ {\rm  on}  \ A\ \ {\rm and} \ \
f({\gamma}')=0, \ {\gamma}'\in \Gamma',
\end{equation}
 or there does not exist
 $g\in V$ satisfying
\begin{equation}g \not\equiv 0 \ \ {\rm  on}  \ A\ \ {\rm and} \ \ g({\gamma})=0, \ {\gamma}\in \Gamma\backslash \Gamma'.
\end{equation}
 \end{theorem}

 The necessity is obvious and the sufficiency
 follows from the following proposition.

 \begin{proposition}\label{finiteset.pr} Let $A$  and $V$ be as in Theorem \ref{AtoGamma.thm}.
 Then there exist a finite set $\Gamma\subset A$ and functions $d_\gamma(x), \gamma\in \Gamma$, such that
 \begin{equation}\label{finiteset.pr.eq1}
|f({x})|^2=\sum_{\gamma\in \Gamma}
d_\gamma({x}) |f({\gamma})|^2, \ x\in A
\end{equation}
 hold for all  $f\in V$.
\end{proposition}

\begin{proof}  Let  $g_n, 1\le n\le N$, be a basis of the space $V|_A$,
and
$W$ be the linear space generated by symmetric matrices
$G({x}):= \big(g_n({x}) g_{n'}({x})\big)_{1\le n, n'\in N}$,  ${x}\in A$.
Then there exists a finite set $\Gamma\subset A$  such that
$G({\gamma}), \gamma\in \Gamma$, is a basis (or a spanning set) for the space $W$.
With the above set $\Gamma$,  we can follow 
 the proof of Theorem \ref{tofinitenew.thm} in Section  \ref{tofinite.section}
   to prove \eqref{finiteset.pr.eq1}.
\end{proof}

 Let  $g_n, 1\le n\le N$, be a basis of the space $V|_A$,  and  $\Gamma$ be as in the proof of Proposition \ref{finiteset.pr}.
  By Theorem \ref{AtoGamma.thm} and  \cite[Theorem 2.8]{BCE06}, we have the following criterion that can be used to verify the
local complement property on a bounded open set $A$ in finite steps.

\begin{theorem}\label{localcomplementcriterion.thm}
The linear space $V$ has the local complement property on $A$ if and only if
  for any $\Gamma'\subset \Gamma$, either $(g_n({\gamma}'))_{1\le n\le N}, {\gamma}'\in \Gamma'$
   form a  frame for $\R^{N}$ or $(g_n({\gamma}))_{1\le n\le N}, {\gamma}\in \Gamma\backslash\Gamma'$
  form a  frame for $\R^{N}$.
\end{theorem}

The local complement property for different open sets can be equivalent. Following the argument used in the proof of Theorem \ref{AtoGamma.thm}, we have

 \begin{proposition}
  Let $A$ be a bounded open set and $V$ be a locally finite-dimensional space with the  local complement property on $A$. If
  $B$ is a bounded open subset of $A$ such that
 signals $g$ and $f$ satisfying $|g({x})|=|f({x})|$ on  $B$
  have same magnitude measurements on $A$, then
         $V$  has local complement property on $B$.
\end{proposition}

The conclusion in the above proposition is not true in general. For instance,
the shift-invariant space $V(\phi_0)$ in Example \ref{hatexample1}
has the local complement property on $(0, 1/2)$, but not on its supset $(0, 1)$.

A linear space may have the local complement property on a bounded open $A$, but not on some of its open subsets.
For instance, one may verify that
 $V(\phi_1)$ has the local complement property on $(0,1)$ and on $(-1/2, 1/2)$, but not on their intersection $(0, 1/2)$, where $\phi_1=\phi_0(2\cdot)$ and $\phi_0$ is given in Example \ref{hatexample1}.  

We finish the appendix with a  proposition about local linear independence
 and local complement property.

\begin{proposition}\label{llilocalcomplement.pr}
Let $\phi$ have local linear independence on any open set. Then there exist $A_m, 1\le m\le M$, such that \eqref{necandsuf.thm.eq1} holds and $V(\phi)$ has the local complement property on $A_m, 1\le m\le M$.
\end{proposition}

\begin{proof} Let $S_k, k\in \Zd$, be as in \eqref{necandsuf.thm.eq2}.
For a set $T\subset \Zd$, define
$S_T=\cap_{k\in T} S_k$. We say that $T$ is maximal if  $S_T\ne \emptyset$ and $S_{T'}=\emptyset$ for all $T'\supsetneqq  T$.  From the definition, there are finitely many maximal sets $T_1, \ldots, T_M$, and denote
the corresponding sets by $A_m:=S_{T_m}, 1\le m\le M$.

Clearly \eqref{necandsuf.thm.eq1} holds for the above selected open sets as
$$\cup_{m=1}^M T_m=\{k\in \Zd:\   S_k\ne \emptyset\}.$$
Then it remains to prove that $V(\phi)$ has local complement property on $A_m, 1\le m\le M$. Assume that $f,g\in V(\phi)$ satisfy $|f(x)|=|g(x)|$  for all $x\in A_m$, which implies that
$(f+g)(x) (f-g)(x)=0$ for all $x\in A_m$. Write $f+g=\sum_{k\in \Zd} c(k) \phi(\cdot-k)$,
$f-g=\sum_{k\in \Zd} d(k) \phi(\cdot-k)$,
 and set $B_1=\{x\in A_m, (f+g)(x)\ne 0\}$ and $B_2=\{x\in A_m: \ (f-g)(x)\ne 0\}$.
Then either $f-g=0$ on $B_1$, or $f+g=0$ on $B_2$, or $f-g=f+g=0$ on $A_m$.
Hence 
either   $c(k)=d(k)$  for all $k\in T_m$ or $c(k)=-d(k)$ on $k\in T_m$
by the local independence on $B_1$, or $ B_2$ or $A_m$.  Therefore either
$f=g$ on $A_m$, or $f=-g$ on $A_m$, or $f=g=0$ on $A_m$. This completes the proof.
\end{proof}

\end{appendix}

\end{document}